%% file: infocom_draft.tex
\documentclass{article}

\usepackage{amsfonts}
\usepackage{epsfig}
\usepackage{amssymb}
\usepackage{amsmath}
\usepackage{color}
\usepackage{graphics}
\usepackage{subfigure}
\usepackage{verbatim} 

\newtheorem{theorem}{Theorem}[section]
\newtheorem{claim}[theorem]{Claim}
\newtheorem{lemma}[theorem]{Lemma}

\newtheorem{proposition}[theorem]{Proposition}

\def\RR{\mathbb{R}}
\def\CF{{\cal F}}
\def\pr{\mbox{P}}
\def\ex{\mbox{E}}
\def\parent{\mbox{\em parent}}
\def\anc{\mbox{\em anc}}

\def\hn{\mbox{\sf HN}}
\def\hnbar{\overline{\hn}}
\def\height{\mbox{\it ht}}
\def\edges{\mbox{\sf EN}}
\def\lev{\mbox{\it lev}}
\def\conn{\mbox{\em conn}}
\def\directory{\mbox{\em directory}}

\def\rec{g}

\def\cone{\vee^\theta}
\def\x{\mbox{\boldmath $x$}}
\def\i{\mbox{\boldmath $i$}}

\newcommand{\ignore}[1]{ }

\newcommand{\sq}{\hbox{\rlap{$\sqcap$}$\sqcup$}}
\newcommand{\qed}{\hspace*{\fill}\sq}

\newenvironment{proof}{\noindent {\bf Proof.}\ }{\qed\par\vskip
  4mm\par}

\def\CF{\mathcal{F} }
\def\CK{\mathcal{K} }

\def\RR{\mathbb{R}}
\def\ZZ{\mathbb{Z}}

\def\N{\mathbb{N}}
\def\direct{\leftrightarrow}

\definecolor{Blue}{rgb}{0.3,0.3,0.9}

\author{
Amitabha Bagchi\footnotemark[1] \and Adit Madan\footnotemark[1] \and Achal Premi\footnotemark[1]
}

\title{Hierarchical neighbor graphs: An energy-efficient bounded degree connected structure for
  wireless networks \footnotetext[1]{Indian Institute of Technology, Hauz Khas, New Delhi, India. Email: bagchi@cse.iitd.ernet.in, madanadit@gmail.com, achalpremi@gmail.com.}
}

\begin{document}
\maketitle

\begin{abstract}
  We introduce hierarchical neighbor graphs, a new architecture for
  connecting ad hoc wireless nodes distributed in a plane. The
  structure has the flavor of hierarchical clustering and requires
  only local knowledge and minimal computation at each node to be
  formed and repaired. Hence, it is a suitable interconnection model
  for an ad hoc wireless sensor network. The structure is able to use
  energy efficiently by reorganizing dynamically when the battery
  power of heavily utilized nodes degrades and is able to achieve
  throughput, energy efficiency and network lifetimes that compare
  favorably with the leading proposals for data collation in sensor
  networks such as LEACH (Heinzelman et. al., 2002). Additionally,
  hierarchical neighbor graphs have low power stretch i.e. the power
  required to connect nodes through the network is a small factor
  higher than the power required to connect them directly. Our
  structure also compares favorably to mathematical structures
  proposed for connecting points in a plane e.g. nearest-neighbor
  graphs~(Ballister et. al., 2005), $\theta$-graphs~(Ruppert and
  Seidel, 1991), in that it has expected constant degree and does not
  require any significant computation or global information to be
  formed.
\end{abstract}


\section{Introduction}
\label{sec:intro}

Topology control is a fundamental problem in the study of wireless ad
hoc and sensor networks. A primary issue in this area is that of
constructing a connected network between the nodes while keeping in
mind the various constraints wireless devices operate under. Several
architectures have been suggested to solve this problem, each striving
to achieve multiple objectives vital to obtaining high throughput and
low latency while expending as little energy as possible. Bounded
degree is one such objective, needed to reduce the overhead of channel
state exchange between neighbors when nodes use MIMO
antennas~\cite{zorzi-twn:2005}.  Another important criterion of a good
solution is that the number of hops between nodes be small, required because
multi-hop wireless networks are error-prone and the probability of
packet loss increases with path length. For wireless nodes with
sufficient battery power, a connected topology must also ensure that
the power stretch--the ratio of power spent by two nodes in connecting
through the network to power spent in communicating directly--must be
kept low~\cite[Chapter 9]{li-cup:2008}.

In this paper we propose a novel architecture for connectivity in
wireless networks: Hierarchical neighbor graphs. Our structure is a
randomized one that effectively combines ideas from skip list data
structures~\cite{pugh-cacm:1991,bagchi-algorithmica:2005} and
nearest-neighbor graphs~\cite{xue-wn:2004,ballister-aap:2005} to give
a hierarchical bounded degree structure for connecting points in a
plane that has short paths between nodes both in terms of number of
hops and distance. Hierarchical neighbor graphs can be built with
local information and without any substantial computation. In order to
find connections, nodes do not need to estimate any angles or
distances, they only need to know the relative distance of their
neighbors, which can be easily determined from signal strength. Unlike
other nearest-neighbor flavored structures, hierarchical neighbor
graphs are able to incorporate battery power as a parameter and ensure
that nodes with low battery power are not expected to transmit to
nodes that are located far away. This property, along with ease of
deployment and reformation as battery power decreases, makes our
structure a good candidate for collating data from a field of wireless
sensor nodes. And, in fact, we demonstrate that our structure is more
energy efficient and has better throughput than one of
the leading proposals in this area, LEACH~\cite{heinzelmann-twc:2002},
while having network lifetime comparable to it.

The key idea of our construction is that each node is assigned a level
and chooses one neighbor from a level above it and its other neighbors
from its own level. This assignment of levels is done through a random
process which helps ensure that the degree of each node is bounded in
expectation no matter what the positions of the points. For the
special case where the locations of the points are generated by a
Poisson point process we show that even when all nodes have equal
battery power, the probability of the edges formed being long is very
small. Encouraged by this we define a edge-length bounded version of
our structure in which all edges above a certain length, determined by
the battery power of the node, are deleted. We demonstrate that for a
Poisson point process of sufficiently high density, this more
practical structure is still connected.

\noindent{\em Paper organization and our contributions.} In
Section~\ref{sec:definition} we define hierarchical neighbor
graphs. We also describe proactive and reactive routing protocols for
this architecture and discuss how it can be easily adapted when nodes
join or leave or when battery power changes. Additionally:
\begin{itemize}
\item We show analytically that hierarchical neighbor graphs have
  bounded degree in expectation (Section~\ref{sec:properties:degree}).
\item When the location of the points are determined by a Poisson
  point process, the probability of long-range connections being
  formed is very low. We also define a bounded connection length
  version of our structure and show through simulation that it is
  connected for Poisson point processes of sufficiently high density
  (Section~\ref{sec:properties:length}).
\item We show analytically that the number of hops between any two
  nodes in our structure is low (Section~\ref{sec:properties:hops}).
\item We show through simulation that the distance stretch of our
  structure has an exponentially decaying distribution, decaying
  faster for pairs of points that are further apart. We show an
  initial theorem in this direction
  (Section~\ref{sec:properties:stretch}). 
\item We simulate the use of our architecture to collate data from a
  field of wireless sensors and show that it performs better than
  LEACH~\cite{heinzelmann-twc:2002} in terms throughput and energy
  efficiency while matching it in terms of network lifetime
  (Section~\ref{sec:sensors}).
\end{itemize}

\subsection{Related work}
\label{sec:intro:related}
The area of topology control for wireless ad hoc and sensor networks
has seen a lot of activity over the years. Several surveys and books
are available on this topic (see
e.g.~\cite{rajaraman-sigact:2002,santi-wiley:2005,li-cup:2008}).  The
proposed architectures range from the geometry-based Gabriel
Graphs~\cite{gabriel-sz:1969} and Relative Neighborhood
Graphs~\cite{toussaint-pr:1980} and their numerous variants, and
direction based proposals like Yao Graphs~\cite{yao-sjc:1982} and
their family. Li gives a fairly comprehensive survey of these kinds of
architectures in Chapter 9 of his book~\cite{li-cup:2008}. Of
particular interest to us in this class of proposals are
$\theta$-graphs~\cite{ruppert-cccg:1991}, in which the area around
each node is partitioned into cones of angle $\theta$ and the node is
connected to the nearest neighbor in each cone. These graphs can be
shown to be spanners i.e. the graph distance between any two points is
a bounded factor larger than their Euclidean distance. Arya, Mount and
Smid~\cite{arya-cgta:1999} refined this proposal in a flavor very
similar to ours by combining skip lists with $\theta$-graphs to
propose skip list spanners, which have the spanning property of
$\theta$-graphs but have fewer hops between all pairs on
nodes--$\theta(\log n)$ hops for a point set with $n$ nodes--than
$\theta$-graphs have. While Arya et. al.'s proposal seems very similar
to ours, they have one major disadvantage they share with all the
other structures of the family of geometric graphs: constructing these
kinds of graphs normally involves computation involving the relative
position of a point and its neighbors and an area surrounding these
points. We note that our hierarchical neighbor graphs avoid the
computation overhead that these kinds of construction require. This
also makes them superior to the class of topology control proposals
that are based on computing dominating sets or spanning trees (see
e.g.~\cite{wightman-globecom:2008}), computations which are time
consuming and construct structures which are difficult to adapt in
dynamic situations. This class of proposals contains several variants,
too numerous to recount here, so we refer the reader to an excellent
survey by Santi~\cite{santi-wiley:2005}.

Clustering based approaches to topology control have also been studied
extensively in the literature.  Some of the common assumptions and
algorithmic features of these proposals are location unawareness,
periodic selection of cluster heads, equal battery power but
multi-power transmission ability.  The systems are evaluated by
evaluating throughput and network lifetime under the assumption that
battery power is not replenished. Within these constraints Younis and
Fahmy proposed a distributed clustering approach, HEED, for maximizing
lifetime of ad hoc sensor networks~\cite{younis-tmc:2004}. This
approach was extended by Huang and Wu, with an additional assumption
of uniform distribution of nodes~\cite{huang:scvt-2005}. Basagni gave
a method of choosing cluster heads based on a node-mobility
parameter~\cite{basagni:ispan-1999}.  Heinzelman, Chandrakasan and
Balakrishnan's architecture, LEACH, is an application specific
protocol for ad hoc wireless
sensors\cite{heinzelmann-twc:2002}. Cluster heads receive data from
sensor nodes and compress it before sending to the base station.
Similar to our approach, LEACH incorporates randomized rotation of
high energy cluster head among sensors, keeping the expected number of
cluster heads in the network constant. By contrast, in our approach
the rotation is implicit and does not require additional
computation. Also, for LEACH, nodes must have an estimate of energy
remaining in the network at the end of each round. Another assumption
LEACH makes that we do not share is that number of nodes in the
network and the number of cluster heads must be known to every
node. In Section~\ref{sec:sensors} we compare our structure to LEACH
and find that under similar assumptions regarding transmission energy,
hierarchical neighbor graphs provide better throughput and energy
efficiency with matching network lifetime.  We note that systems like
LEACH attempt to save power by randomly rotating cluster head
responsibilities. A more sophisticated approach involves choosing
nodes with greater residual energy to perform cluster head roles (See
e.g.~\cite{xiao-acis:2007,younis-tmc:2004}), and this is the paradigm
in which hierarchical neighbor graphs fall.

Another method related to our approach in the sense that clustering is
implicit is CLUSTERPOW introduced by Kawadia and
Kumar~\cite{kawadia-Infocom:2003}. No leader or gateway is explicitly
selected rather the clustered structure of the network is manifested
in the way routing is done. Each node can transmit at different finite
number of power levels. A route is a non increasing sequence of power levels.

On the mathematical front, nearest neighbor models have been studied
by H\"aggstr\"om and Meester~\cite{haggstrom-rsa:1996} who proved that
there is a critical value $k_c$ dependent on the dimension of the
space such that if each point is connected to its $k$ nearest
neighbors for any $k > k_c$, an infinite component exists almost
surely. Restricting this model to a square box of area $n$, Xue and
Kumar~\cite{xue-wn:2004} showed that $k$ must be at least 0.074$\log
n$ for the graph to be connected. This lower bound was improved to
0.3043$\log n$ by Ballister, Bollob\'as, Sarkar and
Walters~\cite{ballister-aap:2005} who also showed that the threshold
for connectivity is at most $0.5139 \log n$ as well as corresponding
results for the directed version of the problem. The advantage our
model enjoys over this model is that by selecting neighbors carefully
from the set of proximate nodes, hierarchical neighbor graphs achieve
constant degree in expectation. Additionally our hierarchical
structure also ensures paths with fewer hops.

Finally, we end by mentioning that hierarchical neighbor graphs bear
more than a passing resemblance to the skip list data
structure~\cite{pugh-cacm:1991} and its biased
version~\cite{bagchi-algorithmica:2005}.

\section{Hierarchical neighbor graphs}
\label{sec:definition}
In this section we define hierarchical neighbor graphs and some
special versions of them. We also describe routing protocols of these
structures and methods for adapting them in dynamic settings. We will
use the notation that for $u,v \in \RR^2$, $d(u,v)$ is the Euclidean
distance between the points. 

Consider a set of points $V \subset \RR^{2}$. We are given a function
$w : V \rightarrow \RR_+$ such that each node $u\in V$ has a battery
power $c \cdot w(u)$ associated with it, where $c$ is a constant
determined by the minimum battery power a node needs to operate
i.e. $c$ ensures that any functioning node has weight at least
1. Taking a parameter $p$ such that $0 < p <1$, we form the
$p$-hierarchical neighbor graph on $V$ with weight function $w$,
denoted $\hn_p^w(V)$ as follows:
\begin{enumerate}
\item We create a sequence $\{S_n : n \geq 0\}$ of subsets of $V$ such
  that $S_0 = V$. The point of a set $S_{i-1}$ are ``promoted'' to
  $S_i$ in two ways, one deterministic and one randomized.
\begin{itemize}
\item Deterministically, all $u \in S_{i-1}$ with $\lfloor
\log_{\frac{1}{p}} w(u) \rfloor \geq i$ are put into $S_i$. 
\item The remaining
points of $S_{i-1}$ are placed in $S_i$ with probability $p$
independently of the choice of all other points. 
\end{itemize}
\item After obtaining the sequence of sets, we say that the {\em
    level} of $u$, $\lev_p(u) = i$ such that $u \in S_i$ and $u \not\in
  S_{i+1}$.
\item Each point $u\in V$ grows a circle around it which stops growing
  the first time a point $v$ with $\lev(v) > i$ is encountered. This
  point is called $\parent(u)$. $u$ makes connections to all nodes $w$
  with $\lev(w) = i$ that lie within this circle and to the node(s) of
  $S_{i+1}$ that lie on the circumference of the circle.  Note that
  this algorithm is fully distributed in nature. All connections are
  made through local interactions.
\end{enumerate}
The
algorithm for forming $\hn^w_p(V)$ is described in
Figure~\ref{fig:formation}.
\begin{figure}[ht]
\begin{center}
\fbox{
\begin{minipage}{0.9\columnwidth}
\noindent{Algorithm {\sf construct}($\hn_p(V)=(V,\edges_p)$)}
\begin{enumerate}
\item $\edges_p=\phi$
\item At each node $u$ do
\begin{enumerate}
\item Determine $\lev_p(u) = i$ such that $u \in S_i$ and
  $u \not\in S_{i+1}$.
\item Find the nearest neighbor $z\in S_{\lev(u)+1}$ of $u$;\\
$\parent(u) \leftarrow z$;
\\$\edges_p^{\lev(u)} \leftarrow \edges_p^{\lev(u)} \cup \{(u, \parent(u))\}$;
\item For each $v\in S_{\lev(u)} \setminus S_{\lev(u)+1}$ such that $d(u,v)\leq d(u,\parent(u))$ i.e
$v$ is closer to $u$ than $\parent(u)$
\\$\edges_p^{\lev(u)} \leftarrow \edges_p^{\lev(u)} \cup \{(u,v)\}$ ;
\end{enumerate}
\item $\edges_p \leftarrow \edges_p \cup \textstyle \bigcup_{i=0}^{\infty} \edges_p^i$.
\end{enumerate}
\end{minipage}}
\caption{Building subgraph $\hn_p(V)$.}
\label{fig:formation}
\end{center}
\end{figure}

For the special case where $w(u) = 1$ for all $u \in V$,
we denote the structure thus formed $\hn_p(V)$. 

\subsubsection{Radius-bounded hierarchical neighbor graphs} We define
a variant of hierarchical neighbor graphs, called radius-bounded hierarchical
neighbor graphs, which take another function $r : V \rightarrow \RR_+$
as a parameter. We denote this graph built on a point set $V$ with
weight function $w$ and parameter $p$ by $\hnbar^{r,w}_p(V)$. The
function $r$ limits the transmission radius of the nodes i.e. node $u$
is not allowed to connect to any node $v$ such that $d(u,v) >
r(u)$. Hence $\hnbar^{w,r}_p(V)$ is a subgraph of $\hn^w_p(V)$ and may
not be connected. This model is, however, a more practical one, since
it does not assume that nodes can form connections with distant
neighbors.  As before, we denote by $\hnbar_p^r(V)$, the special case
of $\hnbar^{w,r}_p(V)$ where all weights are 1.

\subsubsection{Repairing hierarchical neighbor graphs}
Note that hierarchical neighbor graphs are adaptable in the sense that
should a node $u$'s power decreases to $1/p$ times its original amount
its level can be readjusted by reducing the deterministic part of its
promotion by 1, without touching the probabilistic promotions. This
leads to an overall decrease of 1 in $\lev_p(u)$. Similarly if the
power available with a node increases to $1/p$ times the original
amount, the deterministic part of the promotion of a node can be
increased by 1, thereby increasing $\lev_p(u)$ by 1. Clearly when
either of these happens the structure requires some repair which we
discuss next

\paragraph{Node addition or battery power gain} When a node $u$ is
added to $V$, it first determines $\lev_p(u)$. Once its membership in
sets $S_0,\ldots, S_{\lev_p(u)}$ is established, all points $v$ with
$\lev_p(v) < \lev_p(u)$ whose radius of connection contains $u$
truncate their radius of connection and make $u$ their parent. Nodes
$v$ with $\lev_p(v) = \lev_p(u)$ form connections to $u$ if it lies
within their radius of connection. All other nodes are
unaffected. Similarly when a node $u$ gains battery power,
$\lev^w_p(u)$ increases from $i$ to $j$ and nodes of $S_i, S_{i+1},
\ldots, S_j$ may have to truncate their radius of connection if $u$
lies within it.  \paragraph{Node departure or battery power loss} When
a node $u$ departs from $V$, all nodes that have $u$ as their parent
have to expand their radius of connection to find a new parent. Nodes
in $S_{\lev_p(u)}$ that were connected to $u$ simply note its
departure but don't have to change anything. Similarly when the
battery power of a node $u$ in $\hn^w_p(V)$ decreases from $i$ to $j$,
nodes of $S_k \setminus S_{k+1}$ that had $u$ as their parent have to
increase their radius and find a new parent, for $j \leq k \leq
i-1$. Nodes of $S_j \setminus S_{j+1}$ simply note the departure of
$u$ from $S_j$.

\subsection{Routing} We describe two approaches to routing on
$\hn_p(V)$, one proactive and one reactive. Before we do that, let us
define some terms. The ancestor of $u$ at level $i$, $\anc_i(u)$, is
the node of level at least $i$ reachable from $u$ through a series of
edges connecting nodes to their parents. The {\em connected component
  of $u$ at level $i$}, for $i \leq \lev_p(u)$, denoted $\conn_p(u,i)$
is the connected component containing $u$ of the subgraph of
$\hn_p(V)$ comprising only of edges between vertices $v \in V$ such
that $\lev_p(v) \geq i$ i.e. it is the set of vertices whose are
promoted at least up to level $i$ and that are reachable from $u$
without the path having to touch any node of highest level strictly
less than $i$.

\subsubsection{Proactive routing} In this method we extend the established ideas
of distance vector based approaches. The algorithm proceeds in rounds,
one for each level of the hierarchy, consolidating information from
below and passing it up.
\begin{figure}[ht]
\begin{center}
\fbox{
\begin{minipage}{0.9\columnwidth}
\noindent{Algorithm {\sf hierarchicalDistanceVector}($\hn_p(V)=(V,\edges_p)$)}
\begin{enumerate}
\item $i\leftarrow 0$
\item For all $u \in V, 0 \leq j \leq \lev_p(u)$: $\directory_p(u,j) \leftarrow \emptyset$.
\item While there is a node $u$ with $\lev(u) \geq i$ which is not marked visited
\begin{enumerate}
\item Run a Distance Vector algorithm for $\conn_p(u)$ in which each
  node $v$ exchanges $\directory(v,i)$ with all other nodes.
\item For each $v \in \conn_p(u)$ augment $\directory_p(v,i)$ using
  the information gathered from the previous step and
  mark $v$ visited.
\item Send $\directory_p(v,i)$ to $\parent(v)$.
\item Recompute $\directory_p(\parent(v),i)$ by adding the information
  in $\directory_p(v,i)$. 
\item increment $i$
\end{enumerate}
\end{enumerate}
\end{minipage}
}
\caption{Constructing a routing table pro actively in $\hn_p(V)$.}
\label{fig:proactive-routing}
\end{center}
\end{figure}
The protocol is described in Figure~\ref{fig:proactive-routing}.
The idea is that at each level $i$ each node $v$ such that
$\lev_p(v) \geq i$ computes a distance vector to all the nodes in it's
connected component. But along with the distances to a node $v$ in its
component, each node $u$ also receives a distance vector directory of
all nodes of levels lower than $i$ that can be reached through $v$
without visiting any node with highest level $i$ or above. The node
$u$ first uses this information to update its distance vector since it
is possible that a node with a highest level lower than $i$ is better
reached through some other node at level $i$ rather than through links
going down into the hierarchy below $u$. It then consolidates its
distance vector information and sends it to $\parent(u)$. The process
then repeats at $\lev(u)$.

When a packet needs to be sent from $u$ to $v$, $u$ discovers the best
route by checking its directory at level 0. If it doesn't find $v$ listed it
queries $\parent(u)$ which checks the directory at level 1 and so on
till we reach a level $i$ where $\anc_i(u)$'s directory contains an
entry for $v$. The route discovered first goes from $u$ to $\anc_i(u)$
and then proceeds to the node of $S_i$ listed as the next hop for $v$
in $\anc_i(u)$'s directory.

\subsubsection{Reactive routing} The problem with the proactive routing
scheme is that it requires a lot of storage at each level and
computation in the set up phase. There is also the problem of keeping
the routing tables consistent with the current state of the network,
which might be expensive in a situation which evolves quickly. For
situations where storage may not be abundant, e.g. ad hoc sensor
networks, and situations where nodes come and go quickly, e.g. mobile
networks, we now propose a reactive routing scheme which uses a
hierarchical controlled flooding method for route discovery.

The process of discovering a route from $u$ to $v$ proceeds in two
kinds of phases, up the hierarchy and down the hierarchy. At first the
source node $u$, floods $\conn_p(u,0)$ to find if $v$ is in it. If it
doesn't find it, it sends a message to $\parent(u)$. This node now
floods $\conn_p(\parent(u),1)$ asking for a route to $v$. Each node $v
\in \conn_p(\parent(u),1)$ sends messages down to all $w$
s.t. $\parent(w) = v$ asking them to flood their components at level
$0$ in search of $v$. Subsequently, if $v$ is not found in
$\conn_p(\parent(u),1)$, a message is sent to $\parent(\parent(u))$
and it launches a flood at level 2, and so on.
\begin{figure}[htbp]
\begin{center}
\input{reactive-flood.pstex_t}
\caption{If there is an $i$ such that $\anc_i(u) \in \conn(\anc_i(v),i)$
then only the connected components lying at various levels beneath the
vertices of $\conn(\anc_i(v),i)$ get flooded.}
\label{fig:reactive-flood}
\end{center}
\end{figure}
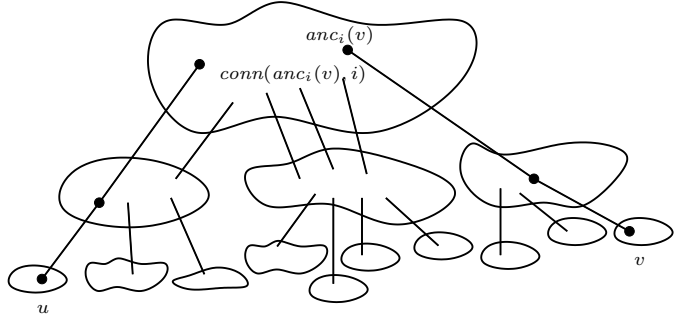
We do not give a formal description of the algorithm here, pointing
instead to Figure~\ref{fig:reactive-flood} to give the reader an idea
of the region of $\hn_p(V)$ that gets flooded in the process of route
discovery between $u$ and $v$. We believe it is possible to argue that
the number of nodes flooded in this process is $o(n)$ but we do not
have a proof for this.


\section{Properties of hierarchical neighbor graphs}
 
\subsection {Bounded degree}
\label{sec:properties:degree}

We first show that for {\em any} point set $V$ the expected degree of the
subgraph $\hn_p(V)$ is $O(\frac{1}{p(1-p)})$. The proof of this can
easily be adapted to show degree bounds for the case where weights are
general, i.e. $\hn^w_p(V)$. We also show that the degree bound for
$\hn_p(V)$ can be improved to $O(\frac{1}{p})$ when $V$ is generated
by a Poisson point process.

\begin{theorem}
\label{thm:det-deg-expectation}
  The expected degree of any point $v\in V$ in $\hn_p(V)$ constructed with
 parameter $p$ on an arbitrary point set $P$ is at most
\[\frac{1}{p} + \frac{6}{p(1-p)}.\]
\end{theorem}
\begin{proof}
Let us consider some point $v \in P$. For the purposes of analysis we
think of the graph as directed and distinguish between {\em outgoing
  edges} i.e. the edges established by $v$ and {\em incoming edges}
i.e. the edges established to $v$ by other points of $P$. We account
for these two sets separately.

\noindent{\em Outgoing edges.} Let us assume that $v$ is promoted up
to level $i$. Clearly $v$ establishes no outgoing edges till this
level. At this level it stops being promoted and searches for the
point nearest to it that is promoted to level $i+1$. Let us name the
points of $S_i$ ordered by increasing distance from $v$ as $u_1, u_2,
\ldots, u_{|S_i|}$. Clearly for $v$ to have $k$ outgoing
edges, the first $k-1$ vertices in this sequence must not be promoted
past level $i$ and the $k$th vertex must be promoted to level
$i+1$. Hence
\[ \pr(v\mbox{ has $k$ outgoing edges}) = (1-p)^k \cdot p.\]
Therefore
\[\ex(\mbox{No. of outgoing edges from } v) \leq \sum_{k=1}^{|S_i|} k\cdot
(1-p)^{k-1} \cdot p \leq \frac{1}{p}.\]

\noindent{\em Incoming edges.} Now let us consider a level $j \leq i$
i.e. a level at which $v$ exists and may be connected to by other
vertices of $P$. We denote by $X_j$ the number of edges established
into $v$ by points that belong to $S_j \setminus S_{j+1}$ i.e. by
points that are promoted up to $S_j$ but not beyond. Note that at
level $j$ the candidates to establish edges to $v$ are looking for a
vertex of level $j+1$ to connect to. These vertices establish an edge
with $v$ either because none of the vertices nearer to them than $v$
have been promoted to level $j+1$.

In order to analyze this situation, let us consider the sequence
$\sigma(v,j) = u_1, u_2, \ldots, u_{|S_j|}$ of vertices of $S_j$
ordered by their increasing distance from $v$. We partition the space
around $v$ into six cones sub tending angles of $\frac{\pi}{3}$ at $v$
and index these six partitions with the numbers 1 to 6. Further we
partition the sequence $\sigma(v,j)$ into six subsequences
$\sigma(v,j,i) = u_1^i, u_2^i, \ldots, 1 \leq i \leq 6$ where the
points of $\sigma(v,j,i)$ lie entirely within cone $i$.  
\begin{figure}[htbp]
\begin{center}
\input{cones.pstex_t}
\caption{Note that $u^1_3$
  cannot have an edge to $u$ if either $u^1_2$ or $u^1_1$ have a
  higher level.}
\label{fig:cones}
\end{center}
\end{figure}
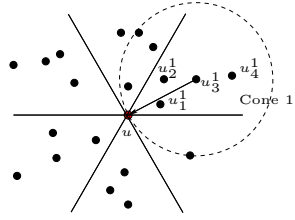
As is clear visually from Figure~\ref{fig:cones}, using elementary
geometry we can claim that the distance between two points within a cone
is less than the distance between the further of these two points and
$v$ i.e.
\begin{claim}
\label{clm:cones}
Given two points $u_k^i$ and $u_\ell^i$ such that $k < \ell$
\[d(u_k,u_\ell) < d(v,u_\ell).\]
\end{claim}
Note that for $v$ to have at least $m$ incoming edges, at least
$\frac{m}{6}$ of these edges must come from points of one of the
six subsequences, say $\sigma(v,j,i^*)$. In view of Claim~\ref{clm:cones}, this
means that the first $k-1$ points of $\sigma(v,j,i^*)$. Hence we can
say that for $v$ to have at least $m$ incoming edges, the first
$\lceil\frac{m}{6}\rceil$ points of $\sigma(v,j,i^*)$ must not be promoted to
level $j+1$ i.e.
\begin{equation}
\label{eq:degree-incoming}
\pr(v\mbox{ has at least }  m \mbox{ incoming edges}) \leq
(1-p)^{\frac{m}{6}}.
\end{equation}
Therefore $\ex(\mbox{No. of incoming edges to }
v \mbox{ at level } j) \leq \sum_{m=1}^{|S_j|}
(1-p)^{\lceil\frac{m}{6}\rceil} \leq 6 \sum_{t=1}^{\frac{|S_j|}{6}}
(1-p)^t \leq \frac{6}{p}$.

What we have shown above is that $\ex(X_j) \leq \frac{6}{p}, j
\geq 0$. Also, we know that for $X_j$ to be non-zero, $v$ must be
promoted at least till level $j$, which happens with probability
$p^j$. With this in view we calculate $\ex(\sum_{j \geq 0} X_j)$ i.e. the
expectation of the  total number of edges incoming to $v$. So, we have 
\begin{eqnarray*}
\ex\left(\sum_{j \geq 0} X_j\right) & = & \sum_{j \geq 0} p^j \cdot \ex(X_j)\\
& \leq & \frac{6}{p(1-p)}.
\end{eqnarray*}
\end{proof}

The proof described above can easily be adapted to show the following
theorem for $\hn^w_p(V)$ for a general weight function:
\begin{theorem}
\label{thm:weighted-det-deg-expectation}
 The expected degree of any point $v\in V$ in $\hn^w_p(V)$ constructed with
 parameter $p$ on an arbitrary point set $V \subset \RR^2$ is at most
\[\frac{1}{p} + \frac{6}{p}\cdot\left(\log_{\frac{1}{p}} w(v)   + \frac{1}{1-p}\right).\]
\end{theorem}
We omit the proof here because it essentially follows the proof of
Theorem~\ref{thm:det-deg-expectation}, pointing out to the reader that
inequality~(\ref{eq:degree-incoming}) holds as a pessimistic
estimation of the probability in the weighted case, since the
deterministic promotion of weighted nodes would make them higher in
general than nodes of weight 1.

For point sets generated by a Poisson point process we show a better
bound of $\theta(\frac{1}{p})$. We note that it is quite expected that
the degree bound does not contain the intensity of the point process
$\lambda$ in it, since $\hn^w_p(V)$ is essentially a nearest neighbor
model.
\begin{theorem}
\label{thm:rand-deg-expectation}
The expectation of the degree of $\hn_p(V)$ constructed with parameter $p$ on
a Poisson point process $V$ with density $\lambda$ is at most $\frac{7}{p}$.
\end{theorem}

\begin{proof}
As before we account separately for outgoing edges and incoming
edges. We inherit the bound for outgoing edges from the proof of
Theorem~\ref{thm:det-deg-expectation} where the setting is more
general. We focus here on improving the bound for incoming edges.

Consider a vertex $v$ that has been promoted up to level $i$. Consider
a level $j\leq i$. We denote by $X_j$ the number of edges established
into $v$ by points that belong to $S_j \setminus S_{j+1}$ i.e. by
points that are promoted up to $S_j$ but not beyond. 

We partition the space around $v$ into six cones sub tending angles
$\frac{\pi}{3}$ at $v$ and number them from 1 to 6. Consider a point
$u$ at a distance $r$ from $v$ in cone $k;1\leq k \leq 6$. Let $X_j^k$
denotes the edges from vertices in cone $k$. By the symmetry of the
cones and linearity of expectation we can say that
\[\ex(X) = 6\cdot\ex(X_1).\] 
So, we focus on one of these cones and compute an upper bound on the
expected number of edges incoming to $v$ from this cone. In view of
Claim~\ref{clm:cones}, for a point $u \in S_j \setminus S_{j+1}$ to
connect to $v$ all the points of $S_j$ in cone 1 which are nearer to
$v$ than $u$ must not be promoted to level $j+1$ i.e. they must all
belong to $S_j \setminus S_{j+1}$. Note that this is an upper bound
since we disregard the points outside the cone which might be closer
to $u$ than $v$ which might be promoted to $S_{j+1}$ and also the
points within the cone but further from $v$ than $u$ which might be
closer of $u$ than $v$ is and which might be promoted to
$S_{j+1}$. Both these kinds of points would prevent $u$ from sending
an edge into $v$ but we disregard them since we are only looking for
an upper bound.

In order to compute this upper bound on $\ex(X_j^1)$ we consider a
segment of the cone at a distance $r$ from $v$ which infinitesimal width
$dr$. Since the area of this strip is infinitesimal, the probability
that there are 2 or more points in this strip is $o(dr)$ and hence can
be neglected. Hence the expected number of edges from this strip into
$v$ can be upper bounded by the probability that there is an edge from
within this strip going into $v$. This is computed by computing the
intersection of the events ``the $i+1$st nearest neighbor of $v$ in
cone 1 is in the strip of width $dr$ at distance $r$'' and ``none of
the $i$ points of $S_j$ in the sector of radius $r$ in cone 1 are
promoted to level $j+1$''. Since these two events happen on disjoint
areas they are independent. To compute the expectation of $X_j^1$ we
simply integrate over all values of $r$ from 0 to $\infty$ for each
value of $i$ and sum over all values of $i$ from 0 to $\infty$ i.e.
\begin{eqnarray*}
  \ex(X_j^1) &\leq&  (1-p) \cdot \int_0^\infty \sum_{i=0}^{\infty} \frac{e^{-\lambda p^j
      \frac{\pi}{6}r^2} \cdot (\lambda p^j\frac{\pi}{6}
    r^2)^i(1-p)^i}{i!}\\
& & \cdot \lambda p^j\frac{\pi}{3}rdr\\
  & = & (1-p) \cdot \int_0^\infty e^{-\lambda p^j \frac{\pi}{6}r^2 p} \cdot\lambda
  p^j \frac{\pi}{3}rdr\\
  & = & \frac{1-p}{p}
\end{eqnarray*}

So we get that
\begin{eqnarray*}
\ex\left(\sum_{j\geq 0}X_j\right) & \leq & 6 \sum_{j\geq 0} p^j\cdot E(X_j^1)\\
& \leq & 6 \cdot \frac{1}{1-p}\cdot \frac{1-p}{p}\\
& = & \frac{6}{p}.
\end{eqnarray*}
\end{proof}

As before, there is a theorem analogous to Theorem~\ref{thm:rand-deg-expectation} for the case where weights are general. We state that theorem here, noting that we do not use the terminology of a Point process in $\RR^2$ since that would have, in general, an infinite number of points and if the weight of each of these points were considered to be at least some $\epsilon >0$, that would amount to an assumption of infinite energy.  \begin{theorem} \label{thm:weighted-rand-deg-expectation}
 The expected degree of any point $v\in V$ in $\hn^w_p(V)$ constructed with
 parameter $p$ on a finite set of points $V$ distributed uniformly at
 random in a bounded area $A \subset \RR^2$ is at most
\[\frac{1}{p} + \frac{6}{p}\cdot\left(\log_{\frac{1}{p}} w(v)   + 1\right).\]
\end{theorem}

The proof of this theorem follows the proof of
Theorem~\ref{thm:weighted-det-deg-expectation} and is omitted here.

\paragraph{Discussion} The result of
Theorem~\ref{thm:rand-deg-expectation} assumes significance when seen
in light of Ballister et. al.'s result on connectivity of $k$-nearest
neighbor graphs~\cite{ballister-aap:2005}. They considered $k$-nearest
neighbor graphs on point sets in $\RR^2$ i.e. each node establishes
edges with its $k$ nearest neighbors. Improving on a result of Xue and
Kumar~\cite{xue-wn:2004}, Ballister et. al. showed that for a Poisson
point process with $\lambda =1$ on a square of area $n$, the
probability of the $k$-nearest neighbor graph being connected tends to
$0$ as $n \rightarrow \infty$ for $k \leq \lfloor 0.3043 \log n
\rfloor$. In contrast, our $p$-hierarchical neighbor graphs achieve
connectivity, even within a square of area $n$ with an expected degree
not depending on $n$. 

Also, we note that nodes in structures like
$\theta$-graphs~\cite{keil-swat:1988,ruppert-cccg:1991} and
Yao~\cite{yao-sjc:1982} graphs have constant outgoing degree but may
have arbitrarily high incoming degree. In fact several papers have
been devoted to constructing constant degree versions of these
structures at some computational expense (see~\cite[Chapter
9]{li-cup:2008} for a thorough
treatment). Theorem~\ref{thm:det-deg-expectation} does not preclude
the possibility of a node having high degree in hierarchical neighbor
graphs, the bound on degree is only in expectation, but it has the
advantage that the kind of positioning of points that prove
pathological for $\theta$-graph-like structures, still retain some
probability of having bounded node degree, at no extra computation
cost.

\subsection{Expected edge length in $\hn_p(V)$}
\label{sec:properties:length}
Since nodes in wireless networks are driven by (limited) power
batteries, their transmission ranges are finite. A network
architecture ignoring this limitation is impractical. In order to
show that hierarchical neighbor graphs are sensitive to this
constraint we will show that the probability of long connections being
formed in $\hn_p(V)$ is very low when $V$ is generated by a Poisson
point process of density $\lambda$. As a consequence of this we expect
that for every value of $r > 0$, the radius-bounded hierarchical
neighbor graph $\hnbar^{w,r}_p(V)$ is connected as long as $\lambda$
is sufficiently high. We demonstrate by simulation that this is indeed
the case. In fact, for every $r > 0$ there is a value
$\lambda_{\min}(r)$ such that $\hnbar^{I,r}_p(V)$ is connected for all
$\lambda > \lambda_{\min}(r)$, for the special weight function $I$
that assigns weight 1 to all nodes.

We begin by bounding the probability that 2 nodes which are distance
$l$ apart are connected directly.
\begin{proposition}
  Consider $\hn_p(V)$ constructed on a point set $V$ generated by a
  Poisson point process of density $\lambda$. Given any two nodes
  $u,v\in V$, with $d(u,v)=l$, the probability that $u$ forms an edge with $v$ or vice
  versa in $\hn_p(V)$, denoted $\rec_{\direct}(l)$, is upper
  bounded as
\[
\rec_{\direct(l)} \leq  \frac{2(1-p)}{(\lambda \pi l^2
  p)^2} \cdot \left\{\frac{1-e^{-\lambda \pi l^2 p} (\lambda \pi
    l^2 p +1)}{log(1/p)}  + \frac{4}{e^2} \right\}.
\]
\end{proposition}
\begin{proof}
 We use $\rec_{\direct}(l)$ to represent this
probability. Consider 2 nodes $u$ and $v$ s.t. $lev_p(u)=i$,
$lev_p(v)=j$ and $d(u,v)=l$. Note that if $i<j$, $u$ connects to $v$
iff there is no node $w$ s.t. $lev_p(w)>i$ and $d(u,w)<d(u,v)$. If
$i=j$, $u$ and $v$ have an edge iff there is no node $w$
s.t. $lev_p(w)>i$ and $d(u,w)<d(u,v)$ or $d(w,v)<d(u,v)$.
\begin{eqnarray*}
\rec_{\direct}(l) & \leq & \sum_{i=0}^\infty \sum_{j\neq i} p^i(1-p)
p^j (1-p) e^{-\lambda \pi l^2 p^{\min\{i,j\}+1}}\\
& &  + \sum_{i=0}^\infty 2 p^{2i}(1-p)^2 e^{-\lambda \pi l^2 p^{i+1}}\\
& = & 2(1-p)\cdot (\sum_{i=0}^\infty  p^{2i} e^{-\lambda \pi l^2 p^{i+1}})
\end{eqnarray*}
Now we upper bound the summation by integrating over the function
$p^{2x} e^{-\lambda \pi l^2 p^{x+1}}$, giving us the result.
\end{proof}

In the radius-bounded hierarchical neighbor graph, $\hnbar_p^{w,r}(V)$
we go a step further towards incorporating real-life constraints and
remove all connections between nodes more than a certain distance
apart. Our simulations showed that for $\hnbar_p^r(V)$ built on a
uniformly and randomly distributed points is still connected for all
$r > 0$, provided the density of these points exceeds a minimum value
that depends on $r$. To investigate the relationship between this
minimum density ($\lambda_{\min}$) and transmission radius ($r$), we
simulated $\hnbar^r_p(V)$ increasing the density of the points
till we achieved a connected network. In order to verify that the
value we determined as $\lambda_{\min}$ was not an anomaly, we ran the
simulation for 10 increments of $\lambda$ past the first value where a
connected network was achieved, and only fixed $\lambda_{\min}$ when
the network was found to be connected for all these 10 increments.
\begin{figure}[htbp]
\begin{center}
\includegraphics[scale=0.8]{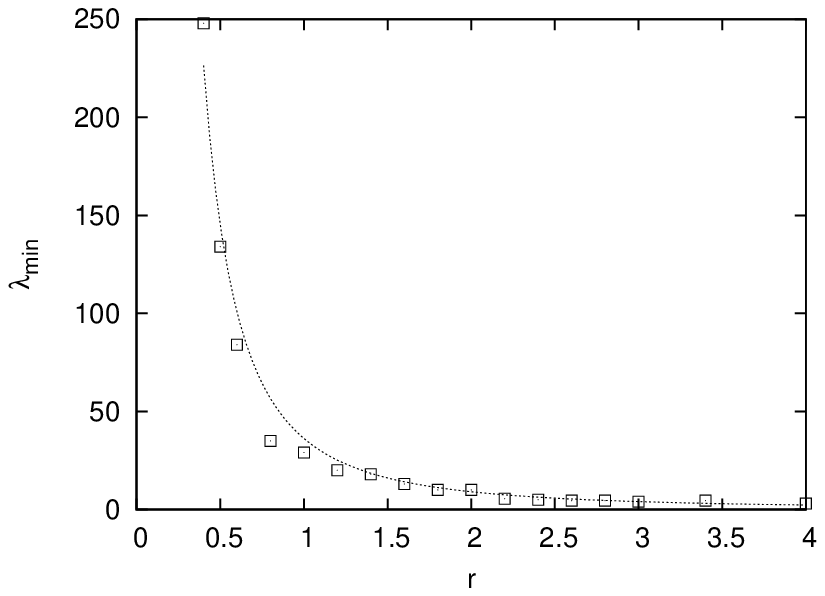}
\caption{For $\hnbar_{0.5}^r(V)$ the curve is $\lambda_{\min} r^2 = 3.62$. (Best fit)}
\label{fig:properties:hnbar}
\end{center}
\end{figure}
Figure~\ref{fig:properties:hnbar} shows the dependency of
$\lambda_{\min}$ on $r$. We fixed $p=0.5$. The point set was scattered
randomly within a $10\times 10$ square. On the Y-axis is
$\lambda_{\min}$ (above which network is connected) plotted against
$r$, the max transmission radius allowed. We found the relation to be
of the form $\lambda_{\min}.r^2 = c$, $c$ being a constant. This is
reminiscent of the critical phenomena of unit disk graphs built on
Poisson point processes. If the radius of each disk is $r$, there is a
critical density $\lambda_c(r)$ above which the unit disk graph has an
infinite component. And it can be shown that $\lambda_c(r).r^2$ is a
constant for $r > 0$ (See \cite{meester:1996} for details). We feel
such a theorem may exist for $\hnbar^r_p(V)$.


\subsection{Number of hops}
\label{sec:properties:hops}
To bound communication delay in networks, one must construct a
topology with a small hop-stretch factor. Hop spanners were introduced
by Peleg and Ullman ~\cite{peleg-podc:1989} and were used as network
synchronizers. In this section we analyze the hop stretch in hierarchical neighbor
graphs. 

For finite point sets and Poisson point processes limited to finite
regions it is easy to show exponential decay of the number of hops
between points. We introduce the notation $\height(\hn_p^W(V)) =
\max_{u\in V} \lev_p(u)$. For finite point sets, we claim the
  following theorem, that follows easily from standard skip list analysis.
\begin{theorem}
\label{thm:properties:height}
For $\hn^w_p(V)$ constructed on a finite set $V$, define $W(V) =
\sum_{u \in V} w(u)$. For $k \geq \max_{u \in V} \log_{\frac{1}{p}} w(u)$,
\[\pr(\height(\hn^w_p(V) \geq k)\leq W\cdot p^k.\] 
\end{theorem}
\begin{proof}
For a given node $u \in V$, the probability that $\lev_p(u) > k$, for
any $k \geq \max_{u \in V} \lfloor \log_{\frac{1}{p}} w(u) \rfloor$ is $p^{k -
  \log_{\frac{1}{p}} w(u)}$. And since the height of $\hn^w_p(V)$ is
at least $k$ if there is at least one node with height at least $k$,
applying the union bound on probabilities we get that
$ \pr(\height(\hn^w_p(V) \geq k) \leq \sum_{u\in V} p^{k -
  \log_{\frac{1}{p}} w(u)} = W\cdot p^k.$
\end{proof}
The nodes at level $\height(\hn^w_p(V))$ are fully connected to each
other (since they are not able to find a node at a higher level). This
means that in $\hn_p(V)$ on a finite set $V$, for any two nodes to connect to
each other the path has to travel at most $O(\height(\hn^w_p(V))$
levels up and down the hierarchy. Hence
Theorem~\ref{thm:properties:height} implies that the number of hops
between any two nodes has an exponentially decaying distribution.

For point processes limited to finite regions, the following
theorem gives us the result that the number of hops between points
decays exponentially:
\begin{theorem}
\label{thm:properties:height-poisson}
Given a finite region $A \subseteq \RR^2$, with area $\ell(A)$, and a
set of points $V$ in this region generated by a Poisson point process
of density $\lambda$, then for any $k \geq 0$,
\[\pr(\height(\hn_p(V) \geq k)\leq \lambda \cdot\ell(A)\cdot p^k.\] 
\end{theorem}
The proof follows easily by conditioning on the number of points in
$A$ and using the argument from the proof of
Theorem~\ref{thm:properties:height} to bound the conditional
probability of the height being greater than $k$. We omit the
details. Theorem~\ref{thm:properties:height-poisson} can be extended
to general weights as well, but we would have to carefully define how
the weights are distributed by the Poisson point process. We omit that
case here since it is basically a mathematical digression.

For Poisson point processes in $\RR^2$ we found through simulation
that the number of hops varies logarithmically with the distance
between the pair of points being connected. We do not have an
anlytical proof for this fact as yet. 

\subsection{Bounding the stretch}
\label{sec:properties:stretch}

A major concern of topology control mechanisms is that the graph be a
spanner i.e. given a point set $V$ and a interconnection structure
$G$, if we denote the shortest distance between points $u,v \in V$
along the edges of $G$ by $d_G(u,v)$, the ratio
\[ \delta_G = \max_{u,v\in V}\frac{d_G(u,v)}{d(u,v)},\] known as the
{\em distance stretch} of $G$ should be low. For example
$\theta$-graphs have distance stretch $\frac{1}{1 -
  \sin(\theta/2)}$~\cite{ruppert-cccg:1991}. The {\em power stretch}
of $G$ is defined as the ratio of the power expended by communicating
through the links of $G$ to the power expended in communicating
directly. The power stretch of $G$ is known to be upper bounded by
$\delta_G^{\beta}$~\cite{li-ccn:2001}, where $\beta$ is a constant
between 2 and 5 that depends on the medium, hence we only consider
distance stretch here.

In order to study the spanner properties of hierarchical neighbor
graphs, we performed a simulation and computed the distance stretch
of
$\hn_p(V)$ constructed on a set of points randomly distributed on a
torus of unit area.
\begin{figure}[htbp]
	\begin{center}
		\includegraphics[scale=0.7 ]{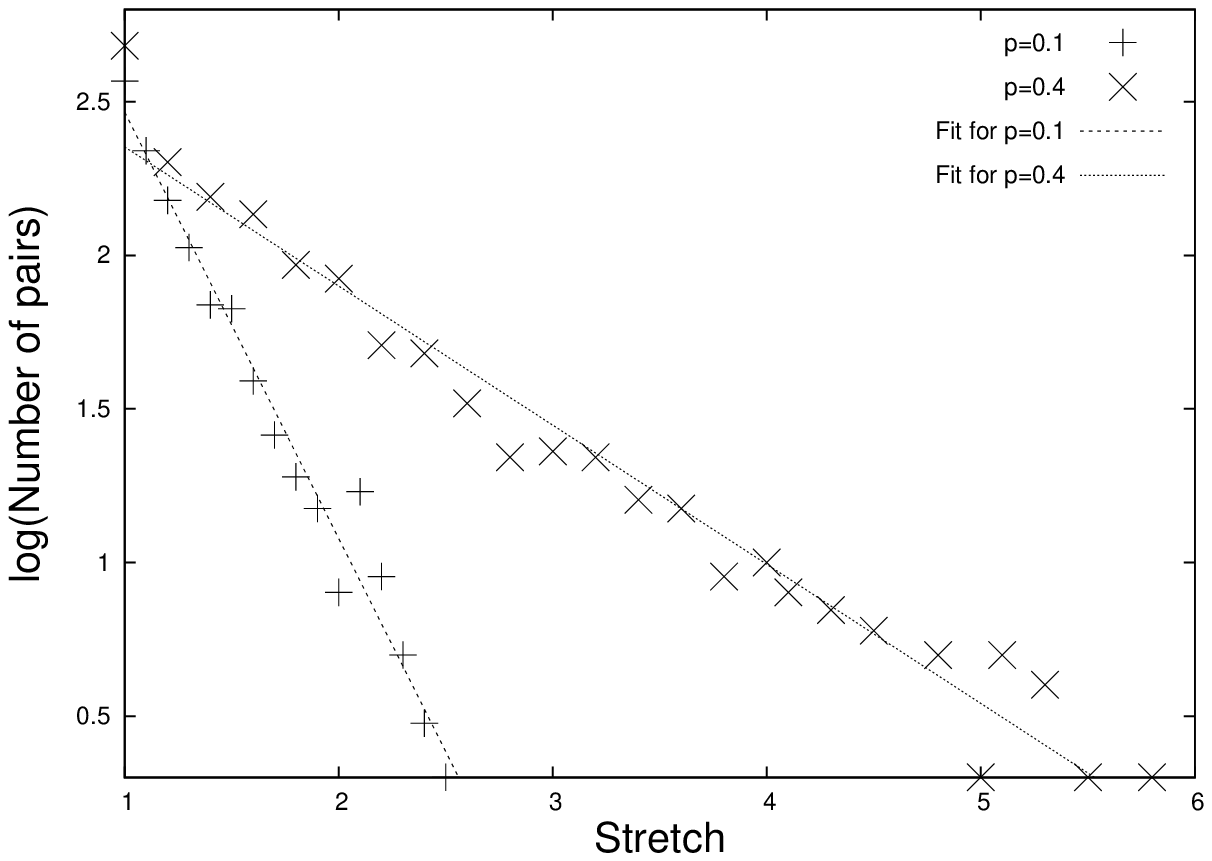}			
		\caption{Stretch graphs for pairs with $d(u,v)=0.1$. $\lambda=500$.}  						
  \label{fig:N-v-stretch}	
\end{center}
\end{figure}
In Figure~\ref{fig:N-v-stretch} we plot the graph distance of pairs of points of $V$ that have a Euclidean distance of 0.1. We found that the log of the number of pairs of vertices stretched to extent $s$ decreases linearly with $s$ which leads us to believe that the probability of a pair of vertices being stretched decays exponentially with the distance stretch value.  We also found that the number of pairs with lower stretch values dominates the number of pairs with higher stretch values as the distance between the pairs increases.  These observations lead us to the conjecture that the stretch for pairs of points in $\hn_p(V)$ follows a probability distribution that decays exponentially in the stretch value, and also decreases as the distance between the pairs increases, which effectively means that distant points undergo very little distortion in their distance when connected through $\hn_p(V)$, although nearby points may occasionally have to connect to each other through a long route.

Analytically, we were able to prove an initial theorem that confirms
our simulation results, but gives a weaker result.
\begin{theorem}
\label{thm:stretch}
Given parameters $p$ and $\lambda$ s.t. $0 < p < 1$ and $\lambda > 0$, 
the graph $\hn_p(V)$ built on a set of points $V$ generated by a
Poisson point process with intensity $\lambda$ has the property that
for any two points $u,v \in V$ such that $d(u,v) = l$, there are positive
constants $c_1$ and $c_2<1$ depending only on $p$ and $\lambda$ such
that for $0 < \theta < \frac{\pi}{3}$,
\[\pr\left(d_p(u,v) > \frac{l}{1 -
  2\sin\left(\frac{\theta}{2}\right)}\right) \\ \leq \exp \left\{
   - \frac{c_1}{l^4 \cdot (1 - c_2 \cdot
      \frac{\theta}{\pi})}\right\}.\]
\end{theorem}

Before proceeding with the proof of Theorem~\ref{thm:stretch}, we need
some preliminaries. Given a point set $V$, our algorithm for
constructing $\hn_p(V)$ defines a probability space $(\Omega, \CF,
\pr)$.  A realization of this process i.e. an element of $\Omega$, is
of the form $\omega = (V=S_0, S_1, \ldots)$. Let us denote the $i$th
set in this tuple as $\omega(i)$.  We define a partial order on the
set of all realizations as follows, $\omega \preceq \omega'$ if
$\omega(i) \subseteq \omega'(i)$ for all $i \geq 0$. An random
variable $N$ defined on this probability space is called an {\em
  increasing} random variable if $N(\omega) \leq N(\omega')$ whenever
$\omega \preceq \omega'$. Similarly a random variable $N$ is called
decreasing if $-N$ is increasing. An event is called increasing if its
indicator is an increasing random variable. For our setting this
definition allows us to state a specific version of the more general
FKG inequality~\cite{fortuin-cmp:1971}.
\begin{lemma}
\label{lem:fkg}
For any two events $A$ and $B$ that are either both increasing or both
decreasing, 
\[ \pr(A \cap B) \geq \pr(A)\cdot\pr(B),\]
where $\pr(\cdot)$ is the probability measure defined on $\Omega$.
\end{lemma}

Now we turn to the main proof.

\begin{proof}
In order to prove this theorem we draw on the paper on $\theta$-graphs
by Ruppert and Seidel~\cite{ruppert-cccg:1991} in which they extended
earlier work by Keil~\cite{keil-swat:1988} and showed that paths that
stayed within a cone of angle $\theta$ centered around the line joining
a point to its destination have stretch at most
$1/(1-2\sin(\theta/2))$. The proof proceeds by calculating a lower
bound on the probability that a there is a path between any two points
that stays within a cone of angular width $\theta$.

Given two points $u$ and $v$, we define the {\em $\theta$-cone of $u$
  w.r.t $v$} as a sector of a circle of radius $d(u,v)$ centered at
$u$ with an angle $\theta$ which is bisected by the line segment
joining $u$ and $v$. We denote this cone $\cone_v(u)$.  We define a
{\em $\theta$-good path} between vertices $u$ and $v$ as a path
$u=u_0\rightarrow u_1 \ldots \rightarrow u_n=v$ such that $u_{i+1}$ is
in $\cone_v(u_i)$. 

We also define events that will be relevant to us
\begin{itemize}
\item $G_{uv}^{\theta}$: There is a $\theta$-good path between vertices $u$ and $v$.

\item $G_{uv}^{\theta}(i,j)$: There is a $\theta$-good
  path between vertices $u$ and $v$ conditioned on the fact that
  $\lev(u) = i$ and $\lev(v) = j$.

\item $G_{u\direct v}$: The event that there is a direct edge between vertices $u$ and $v$.

\item $G_{u\direct v}(i,j)$: The event that there is a direct edge
  between vertices $u$ and $v$ conditioned on the fact that 
  $\lev(u) = i$ and $\lev(v) = j$. 
\end{itemize}

Our proof proceeds by computing the probability that the two points
are either directly connected or $u$ is connected to a point $w$ in
$\cone_v(u)$ which then has a $\theta$-good path joining it to
$v$. Note that the second case is an intersection event. Hence in
order to proceed we will need the following lemma which shows that the
two events being intersected are positively correlated. 

\begin{lemma}
\label{lem:fkg-application}
Given a set of points $X$ generated by a Poisson point process in
$\RR^2$, consider three points $u,v,w \in X$ such that $w \in \cone_v(u)$. Then
\[ \pr \left(\bigcup_{k=0}^{\infty} \{G_{u \direct
  w}(i,k) \cap G^\theta_{wv}(k,j)\}\right) \geq \sum_{k=0}^\infty 
\pr \left( G_{u \direct w}(i,k)\right) \cdot \pr\left( G^\theta_{wv}(k,j)\right).\]
\end{lemma}

\begin{proof} We proceed by showing
  that for a given realization of the Poisson point process the events
  we are considering are both decreasing. Hence we can apply the FKG
  inequality to these events conditioned on a particular realization, 
  then decondition to get the result. Let us formalize this.

  Following~\cite{meester:1996} we view the Poisson point process $C$ in
  the following way: Consider the family of boxes
  \[ K(n,z=(z_1,z_2)) = \left(\frac{z_1}{2^n}, \frac{z_1+1}{2^n}\right] \times
  \left(\frac{z_2}{2^n}, \frac{z_2+1}{2^n}\right], \mbox{ for all } n \in \N, z
  \in \ZZ^2,\]
For a given realization of $X$, each point $x \in X$ is contained in a
unique box $K(n,z(n,x))$ for each $n \in \N$. With probability 1 there
is a unique smallest number $n(x)$ such that $K(n(x), z(n(x),x))$
contains no other points of $X$. Hence we can view the realizations of
$X$ in terms of sets of boxes taken from
\[\CK(\ZZ^2) = \bigcup_{n \in \N} \bigcup_{z \in \ZZ^2} K(n,z),\]
such that each box contains exactly one point i.e. we have a
measurable mapping from the probability space of the Poisson point
process to a space $\Omega_1 \subseteq \CK(\ZZ^2)$ and an outcome
$\omega \in \Omega_1$ is a set of boxes containing exactly one point
each.  

In order to describe the probability space on which $\hn_p(X)$ is
defined, we observe that each realization of $\hn_p(X)$ can be
completely described by the levels of the points of a particular
realization of $X$. Defining $\Omega_2 = \N \cup \{0\}$ we get that
$\hn_p(X)$ is defined on a product space $\Omega = \Omega_1 \times
\Omega_2$.

Now, let us consider a given $\omega \in \Omega_1$ and two points $u,v
\in X$ for this outcome. Let us denote by $G(\omega)_{u \direct
  v}(i,j)$ the event that for the realizations mapped to $\omega$, $u$
and $v$ are connected by a directed edge, conditioned on the fact that
$\lev(u) = i$ and $\lev(v) = j$. We claim that this is a decreasing
event. To see this note that the event $G(\omega)_{u \direct v}(i,j)$
depends only on the levels of all the points of $\omega$ apart from
$u$ and $v$. Consider two outcomes $\alpha$ and $\alpha'$ contained in
the subspace of $\Omega$ defined by $\omega_1$ and
$\{\lev(u)=i,\lev(v)=j\}$. Assume that $\alpha \preceq \alpha'$ as
defined at the beginning of this section.

Let us assume that $i \leq j$. Clearly if $G(\omega)_{u \direct
  v}(i,j)$ does not occur in $\alpha$, there is some point $w \in
\omega_1$ such that $\lev(w)(\alpha)$ i.e. the value of $\lev(w)$ in
outcome $\alpha$, is strictly greater than $i$ and $d(u,w) \leq
d(u,v)$. Since $\alpha \preceq \alpha'$, this point $w$ exists in
$\alpha'$ and $\lev(w)(\alpha') \geq \lev(w)(\alpha) > i$ and hence
there cannot be a direct edge from $u$ to $v$ in $\alpha'$ either. If
$i = j$, we can find two such points that prevent a direct edge being
formed between $u$ and $v$ and exactly the same argument holds. We
state this as a proposition.
\begin{proposition}
\label{prp:direct}
Given an $\omega \in \Omega_1$ and two points $u, v$ of the Poisson
point process, $G(\omega)_{u \direct v}(i,j)$ is a decreasing event for
all integers $i,j \geq 0$.
\end{proposition}
Again starting with an $\omega \in \Omega_1$ and a sequence of points
$\x = x_1, \ldots, x_n, n\geq 2$ of a realization of the Poisson
point process mapped to $\omega$, we define the event
\[ A(\omega)_{\x} :  G_{x_1 \direct x_2} \cap \cdots \cap G_{x_{n-1} \direct x_n},\]
and given another sequence $\i = i_1, \ldots, i_n, n \geq
2$ of non-negative integers, we define
\[ A(\omega)_{\x}(\i) = \{ A(\omega)_{\x} \mid \lev(x_j) = i_j, 1 \leq j \leq n\}.\]
In other words, the event $A(\omega)_{\x}(\i)$ is the event that the points
of $\x$ form a path in $\hn_p(X)$ conditioned on their levels being
fixed. We claim that events of the form $A(w)_{\x}(\i)$ are decreasing
events.
\begin{proposition}
\label{prp:paths}
Given an $\omega \in \Omega_1$ and a sequence of points
$\x = x_1, \ldots, x_n, n\geq 2$ of the point process, $A(\omega)_{\x}(\i)$
is a decreasing event for every sequence $\i = i_1, \ldots, i_n, n \geq
2$ of non-negative integers.
\end{proposition}
This is not difficult to see since the event $A(w)_{\x}(\i)$ is a
finite intersection of the kind of events we argued were decreasing in
Proposition~\ref{prp:direct}.

Before proceeding we observe that a bounded area around $\cone_v(u)$
is all that we are considering since points outside a certain region
do not affect the events we are discussing. In the following when we
talk about an outcome $\omega \in \Omega_1$, we will in fact only be
referring to that outcome limited within this bounded region. 

Now we turn to the left hand side of the inequality in the statement
of the lemma and note that it can be rewritten as
\begin{equation}
\label{eq:fkg-lhs}
\sum_{\omega \in
  \Omega_1} \pr \left(\bigcup_{k=0}^{\infty} \{G(\omega)_{u \direct
  w}(i,k) \cap G(\omega)^\theta_{wv}(k,j)\}\right)\cdot \pr(\omega),
\end{equation}
where $\pr(\omega)$ is the probability of a particular outcome. Note
that since we have restricted our outcomes to a finite region
containing $\cone_v(u)$, the probability $\pr(\omega)$ is non-zero.

We observe that the event $G(\omega)^\theta_{wv}(k,j)$ can be written
as a union of events of the form $A(\omega)_{\x}(\i)$ where $x_1 = w$
and $x_n = v$ as long as all the points of $\x$ have the property of
$\theta$-good paths. To write this formally we define $\sigma^n(\omega)(w,v)$ to be all the
sequences $\x$ of points that have these properties i.e. $x_1 = w, x_n
= v, x_i \in \cone_v(x_{i-1}), 2 \leq i \leq n$. Also we use the
notation $\N_0^i = \prod_{j=1}^i \N\cup\{0\}$. Hence we get
\[ G(\omega)^{\theta}_{wv}(k,j) = \bigcup_{n=2}^{\infty}\ \ \ \bigcup_{\i
  \in \N_0^{n-2}} \ \ \bigcup_{\x \in \sigma^n(\omega)(w,v)}
A(\omega)_{\x}(\i). \]
Putting this into~(\ref{eq:fkg-lhs}) we get that 
\begin{eqnarray*}
\mbox{LHS}&  = & \sum_{\omega \in
  \Omega_1} \pr \left(\bigcup_{k=0}^{\infty}\  \bigcup_{n=2}^{\infty}\ \ \ \bigcup_{\i
  \in \N_0^{n-2}} \ \ \bigcup_{\x \in \sigma^n(\omega)(w,v)} \{G(\omega)_{u \direct
  w}(i,k) \cap A(\omega)_{\x}(\i) \}\right)\cdot \pr(\omega)\\
& = & \sum_{\omega \in
  \Omega_1} \sum_{k=0}^{\infty}  \sum_{n=2}^{\infty}\ \ \sum_{\i
  \in \N_0^{n-2}}\ \  \sum_{\x \in \sigma^n(\omega)(w,v)} \pr \left( \{G(\omega)_{u \direct
  w}(i,k) \cap A(\omega)_{\x}(\i) \}\right)\\ 
& & \cdot \pr(\omega)\cdot
\pr(\lev(w) = k) \cdot \pr(\lev(x_j) = i_j, 2\leq j
\leq n-1).
\end{eqnarray*}
From Propositions~\ref{prp:direct} and~\ref{prp:paths} and the FKG
inequality (Lemma~\ref{lem:fkg}) we get that
\begin{eqnarray*} 
\mbox{LHS} & \geq & \sum_{\omega \in
  \Omega_1} \sum_{k=0}^{\infty}  \sum_{n=2}^{\infty}\ \ \sum_{\i
  \in \N_0^{n-2}}\ \  \sum_{\x \in \sigma^n(\omega)(w,v)} \pr \left(G(\omega)_{u \direct
  w}(i,k)\right)\cdot\pr\left( A(\omega)_{\x}(\i) \}\right)\\ 
& & \cdot \pr(\omega)\cdot
\pr(\lev(w) = k) \cdot \pr(\lev(x_j) = i_j, 2\leq j
\leq n-1)\\
& \geq & \sum_{\omega \in
  \Omega_1} \sum_{k=0}^{\infty} \pr \left(G(\omega)_{u \direct
  w}(i,k)\right)\cdot\pr(\lev(w) = k) \cdot \pr(\omega) \\ 
& & \cdot \sum_{\omega \in
  \Omega_1} \sum_{n=2}^{\infty}\ \ \sum_{\i
  \in \N_0^{n-2}}\ \  \sum_{\x \in \sigma^n(\omega)(w,v)} \pr\left(
A(\omega)_{\x}(\i) \}\right) \cdot \pr(\omega)\cdot 
 \pr(\lev(x_j) = i_j, 2\leq j
\leq n-1)\\
& = & \mbox{RHS}
\end{eqnarray*}
\end{proof}

Since the underlying point set $V$ is produced by a stationary
process, the probability of having a particular path between a pair of
points a certain distance apart is the same as the probability of
having the same kind of path for another pair which is the same
distance apart. Hence we define the following functions which we will need
\begin{itemize}
\item $\rec(l) = \pr(G_{uv}^{\theta} \mid d(u,v)=l)$. 

\item $\rec(i,j,l) = \pr(G_{uv}^{\theta} \mid , d(u,v)=l), lev_p(u)=i, lev_p(v)=j)$.

\item $\rec_{\direct}(l) = \pr(G_{u\direct v} \mid d(u,v)=l)$. 

\item $\rec_{\direct}(i,j,l) = \pr(G_{u\direct v} \mid lev_p(u)=i, lev_p(v)=j, d(u,v)=l)$.

\end{itemize}

We start by looking at $\rec(i,j,l)$ for some $i,j >0$. We divide this
probability into two cases. The first being that there is a direct
edge between $u$ and $v$, and the second that there is a direct edge
between $u$ and a point $w$ of level $k$ which further has a $\theta$-good
path to $v$. We sum this probability over all possible values of
$k$. Note that when $i\leq k$, $u$ connects to $w$ when looking for a
point of level $i+1$, and when $i\geq k$, $w$ connects to $u$ when
looking for a point of level $k+1$. In case $w$ gives an incoming edge
to $u$, we sum over the probability that the closest neighbor of $u$
belonging to level $k$ has a $\theta$-good path to $v$. We get that
\begin{eqnarray*} 
\pr(G_{uv}^{\theta}(i,j)) & = & \pr(G_{u\direct v}(i,j)) + (1-\pr(G_{u \direct
  v}(i,j)))\cdot \pr\left(\bigcup_{w \in \cone_v(u)}  G_{u\direct w} \cap
G_{wv}^{\theta}\right)\\  
& = & \pr(G_{u\direct v}(i,j)) + (1-\pr(G_{u \direct v}))\cdot
\pr\left(\bigcup_{w \in \cone_v(u)}\bigcup_{k \geq 0} G_{uw}(i,k) \cap G_{wv}^{\theta}\right) \\  
& = & \pr(G_{u\direct v}) + (1-\pr(G_{u \direct v}))\\ 
& &  \cdot
\int_0^{d(u,v)} \pr\left(\bigcup_{w \in \cone_v(u)}\bigcup_{k \geq 0} \{G_{uw}(i,k)
\cap G_{wv}^{\theta} \mid x \leq d(u,w) \leq x+dx\}\right)\\
&&  \cdot\ \pr(\exists w
\in \cone_v(u), x \leq d(u,w) \leq x+dx).
\end{eqnarray*}
Now since we are considering infinitesimal widths, the probability
that there is more than one point at a distance between $x$ and $x+dx$
from $u$ in $\cone_v(u)$ is $o(dx)$ so we can neglect it. Hence we get 
\begin{eqnarray*}
\pr(G_{uv}^{\theta}(i,j)) & = &\int_0^{d(u,v)} \pr\left(\bigcup_{k \geq 0} \{G_{uw}(i,k)
\cap G_{wv}^{\theta} \mid x \leq d(u,w) \leq x+dx\}\right)  \cdot
\lambda \theta x dx.
\end{eqnarray*} 
Now, using the result of Lemma~\ref{lem:fkg-application} we get 
\begin{eqnarray*}
\pr(G_{uv}^{\theta}(i,j)) & \geq & \sum_{k=0}^\infty \int_0^{d(u,v)}
\pr \left( G_{u \direct w}(i,k)\mid x \leq d(u,w) \leq x+dx\right)\\
& & \cdot\ \pr\left( G^\theta_{wv}(k,j) \mid x \leq d(u,w) \leq x+dx \right) \cdot
\lambda \theta x dx.
\end{eqnarray*}

For $k \leq i$ we have
\[\pr(G_{uv}^{\theta}(i,j)) \geq
\sum_{k=0}^\infty \int_0^{d(u,v)}e^{-\lambda \pi x^2 p^k} \cdot e^{-\lambda \pi x^2 p^{k+1}} \lambda
\theta x dx.\]
and for $k > i$
\[\pr(G_{uv}^{\theta}(i,j)) \geq
\sum_{k=0}^\infty \int_0^{d(u,v)}e^{-\lambda \pi x^2 p^k} \cdot e^{-\lambda \pi x^2 p^{i+1}} \lambda
\theta x dx.\] 

Let us assume that $d(u,v) = l$ 
In the previous equation if $d(w,v)= \delta$ we get
\begin{eqnarray*} 
\rec(i,j,l) & \geq  &\rec_{\direct}(i,j, l)\\ 
 & & + (1-\rec_{\direct}(i,j,
l)) \cdot \sum_{k=0}^{i}\left(
  \int_0^l e^{-\lambda \pi x^2 p^k} \cdot e^{-\lambda \pi x^2
    p^{k+1}} \lambda \theta x dx \right) \cdot p^k(1-p) \cdot
\rec(k,j,l - \delta) \\
 & & +  (1-\rec_{\direct}(i,j, l)) \cdot \sum_{k=i+1}^\infty \left( \int_0^l
e^{-\lambda \pi x^2 p^k} \cdot e^{-\lambda \pi x^2 p^{i+1}} \lambda \theta x dx\right) \cdot p^k(1-p) \cdot
\rec(k,j,l - \delta).
\end{eqnarray*}
For $n > m$
\[ e^{-\lambda \pi x^2 p^n} > e^{-\lambda \pi x^2 p^m}.\]
Putting these values in we get 
\begin{eqnarray*} \rec(i,j,l) & \geq & \rec_{\direct}(i,j,l) + (1-\rec_{\direct}(i,j,l)) \cdot \left\{\sum_{k=0}^{i}\left(
  \int_0^l e^{-2\lambda \pi x^2 p^{k}} \lambda \theta x dx \right) \cdot p^k(1-p) \cdot
\rec(k,j,l - \delta) \right.\\
 & & \left. +  \sum_{k=i+1}^\infty \left( \int_0^l
 e^{-2 \lambda \pi x^2 p^{i+1}} \lambda \theta x dx\right) \cdot p^k(1-p) \cdot
\rec(k,j,l - \delta)\right\} \\
& \geq  & \rec_{\direct}(i,j,l) +
   (1-\rec_{\direct}(i,j,l)) \cdot \left\{\frac{\theta}{4 \pi} \cdot\sum_{k=0}^{i}\left(\frac{1 - e^{-2 \pi l^2 p^{k}}}{p^{k}}
 \right) \cdot p^k(1-p) \cdot
\rec(k,j,l - \delta) \right. \\
 & & \left. + \frac{\theta}{4 \pi}\cdot  \left( \frac{1 - e^{-2 \pi l^2 p^{i+1}}}{p^{i+1}}
\right) \cdot \sum_{k=i+1}^\infty p^k(1-p) \cdot
\rec(k,j,l - \delta)\right\}.
\end{eqnarray*}
We use the simple fact that for $k \leq i$
\[  (1 -
  e^{-2\lambda \pi l^2 p^k})  \geq (1 -
  e^{-2\lambda \pi l^2 p^i}),\]
to get
\begin{eqnarray*} \rec(i,j,l) & \geq  &\rec_{\direct}(i,j,l) +
(1 - \rec_{\direct}(i,j,l)) \cdot \left\{
   \frac{\theta}{4 \pi} \cdot  (1- e^{-2\lambda \pi l^2 p^i}) \cdot
   \sum_{k=0}^{i}\left(\frac{1}{p^{k}} 
 \right) \cdot p^k(1-p) \cdot
\rec(k,j,l - \delta) \right.\\
 & & \left. + \frac{\theta}{4 \pi}\cdot  \left( \frac{1 - e^{-2 \lambda \pi l^2 p^{i+1}}}{p^{i+1}}
\right) \cdot \sum_{k=i+1}^\infty p^k(1-p) \cdot
\rec(k,j,l - \delta)\right\}.
\end{eqnarray*}

In order to simplify this we prove a small lemma:
\begin{lemma}
\label{lem:exp-sums}
Given a parameter $p$ s.t. $0 < p < 1$ and two sequences, $\{a_n =
\frac{1}{p^n} : n \geq 0\}$ and $\{b_n : n \geq 0\}$ which takes
non-negative values:
\[ \sum_{n=0}^i a_n \cdot b_n \geq \frac{p^i(1-p)}{1 - p^{i+1}} \cdot
\left(\sum_{n=0}^i a_n\right) \cdot \left( \sum_{m=0}^i b_m\right). \]
\end{lemma}

\begin{proof}
We start from the right hand side of the inequality:
\begin{eqnarray*}
\sum_{n=0}^i a_n \cdot  \sum_{n=0}^i b_n & =
& \sum_{n=0}^i a_n \cdot b_n + \sum_{n=0}^i b_n \cdot \sum_{m=0,m\ne
  n}^i a_m\\
& \leq & \sum_{n=0}^i a_n \cdot b_n + \sum_{n=0}^i b_n \cdot \sum_{m=1}^i a_m\\
& = & \sum_{n=0}^i a_n \cdot b_n + \sum_{n=0}^i b_n \cdot \sum_{m=1}^i \frac{1}{p^m}\\
& \leq & \sum_{n=0}^i a_n \cdot b_n + \sum_{n=0}^i b_n \cdot \frac{1-p^i}{(1-p)p^i}\\
\end{eqnarray*}
Now, since $\{a_n\}$ is an increasing sequence whose minimum value is $1$,
achieved at $n=0$ we can say that

\begin{eqnarray*}
\sum_{n=0}^i a_n \cdot  \sum_{n=0}^i b_n & \leq & \sum_{n=0}^i a_n \cdot b_n + \frac{1-p^i}{p^i (1-p)} \cdot
\sum_{n=0}^i b_n \cdot a_n\\
& \leq & \frac{1-p^i+p^i(1-p)}{p^i (1-p)} \cdot
\sum_{n=0}^i b_n \cdot a_n.
\end{eqnarray*}
\end{proof}

Using the inequality of Lemma~\ref{lem:exp-sums} we get
\begin{eqnarray*} \rec(i,j,l) & \geq  & \rec_{\direct}(i,j,l) 
 +
  \left\{\frac{\theta}{4 \pi} \cdot  (1- e^{-2\lambda \pi l^2 p^i}) \cdot \frac{p^i(1-p)}{1 - p^{i+1}} \sum_{k=0}^{i}\left(\frac{1}{p^{k}}
 \right) \cdot \sum_{k=0}^i p^k(1-p) \cdot
\rec(k,j,l - \delta) \right.\\
 & & \left.+ \frac{\theta}{4 \pi}\cdot  \left( \frac{1 - e^{-2\lambda \pi l^2 p^{i+1}}}{p^{i+1}}
\right) \cdot \sum_{k=i+1}^\infty p^k(1-p) \cdot
\rec(k,j,l - \delta)\right\} \cdot (1-\rec_{\direct}(i,j,l))\\
& \geq & \rec_{\direct}(i,j,l) 
 +
  \left\{ \frac{\theta}{4 \pi} \cdot  (1- e^{-2\lambda \pi l^2 p^{i+1}}) \cdot
  \frac{1-p^{i+1}}{1 - p^{i+1}} \cdot \sum_{k=0}^i p^k(1-p) \cdot
\rec(k,j,l - \delta) \right.\\
 & & \left.+ \frac{\theta}{4 \pi}\cdot  \left( \frac{1 - e^{-2\lambda \pi l^2 p^{i+1}}}{p^{i+1}}
\right) \cdot \sum_{k=i+1}^\infty p^k(1-p) \cdot
\rec(k,j,l - \delta) \right\} \cdot (1-\rec_{\direct}(i,j,l))\\
\end{eqnarray*}
We now get a form which is easier to handle i.e.
\begin{equation} 
\label{eq:i-j-l}
\rec(i,j,l) \geq   \rec_{\direct}(i,j,l) 
 + \cdot (1-\rec_{\direct}(i,j,l)) \frac{\theta}{4 \pi} \left(1 - e^{-2\lambda \pi l^2 p^{i+1}}\right)  
 \cdot \sum_{k=0}^\infty p^k(1-p) 
\rec(k,j,l - \delta).
\end{equation}
Observing that $\rec(l) = \sum_{i=0}^\infty \sum_{j=0}^\infty p^i (1-p) p^j (1-p)
\rec(i,j,l)$, and noting that for $i \leq j$,
\[\rec_{\direct}(i,j,l) = e^{-\lambda \pi l^2 p^{i+1}},\]
the inequality of~(\ref{eq:i-j-l}) gives us
\begin{eqnarray*} \rec(l) & \geq &\sum_{i=0}^\infty
  \sum_{j=0}^\infty p^i (1-p) p^j (1-p) \rec_{\direct}(i,j,l) \\ & & + 
  \sum_{i=0}^\infty  p^i (1-p) \frac{\theta}{4 \pi} \cdot \left(1 -
    e^{-2\lambda \pi l^2 p^{i+1}}\right) \cdot\sum_{j=0}^{i-1} \left(1 -
    e^{-\lambda \pi l^2 p^{j+1}}\right) \cdot
\sum_{k=0}^\infty p^k(1-p) p^j (1-p) \cdot
\rec(k,j,l - \delta)\\
& & 
 +
  \sum_{i=0}^\infty  p^i (1-p) \frac{\theta}{4 \pi} \cdot \left(1 -
    e^{-2\lambda \pi l^2 p^{i+1}}\right) 
\cdot \left(1 -
    e^{-\lambda \pi l^2 p^{i+1}}\right) \cdot\sum_{j=i}^\infty 
\sum_{k=0}^\infty p^k(1-p) p^j (1-p) \cdot
\rec(k,j,l - \delta)\\
\end{eqnarray*}

Observe that 
\[\rec_{\direct}(l) = \sum_{i=0}^\infty \sum_{j=0}^\infty p^i (1-p)
p^j (1-p) \rec_{\direct}(i,j,l). \]
With this and the simple fact that for $j < i$
\[  (1 -
  e^{-\lambda \pi l^2 p^j})  > (1 -
  e^{-\lambda \pi l^2 p^i}),\]
we get
\begin{eqnarray*}
\rec(l) & \geq & \rec_{\direct}(l) +  \sum_{i=0}^\infty  p^i (1-p) \frac{\theta}{4 \pi} \cdot \left(1 -
    e^{-2\lambda \pi l^2 p^{i+1}}\right) \cdot \left(1 -
    e^{-\lambda \pi l^2 p^{i+1}}\right) \cdot \rec(l - \delta).
\end{eqnarray*}
We begin by lower bounding $\rec_{\direct}(l)$.
\begin{eqnarray*}
\rec_{\direct}(l) & = & \sum_{i=0}^\infty \sum_{j=0}^\infty p^i p^j (1-p)^2 e^{-\lambda \pi l^2 p^{\min\{i,j\}+1}} \\
& = & 2 \sum_{i=0}^\infty \sum_{j=i}^\infty p^i p^j (1-p)^2
e^{-\lambda \pi l^2 p^{i+1}} - \sum_{i=0}^\infty  p^{2i} (1-p)^2
e^{-\lambda \pi l^2 p^{i+1}}\\ 
& = & (1-p^2)\cdot (\sum_{i=0}^\infty  p^{2i} e^{-\lambda \pi l^2 p^{i+1}})
\end{eqnarray*}
Now we lower bound the summation by integrating over the decreasing
part of the function $p^{2x} e^{-\lambda \pi l^2 p^{x+1}}$, after
which we get 
\begin{equation}
\label{eq:direct}
\rec_{\direct}(l) \geq (1-p^2) \cdot (\sum_{i=0}^\infty  p^{2i}
e^{-\lambda \pi l^2 p^{i+1}}) \geq \frac{(1-3e^{-2})(1-p^2)}{(\lambda
  \pi l^2 p)^2}.
\end{equation}

We now turn to  the coefficient of $\rec(l - \delta)$ and observe that, 
\begin{eqnarray*}
\sum_{i=0}^\infty  p^i \cdot \left(1 -
    e^{-2\lambda \pi l^2 p p^i}\right) \left(1 -
    e^{-\lambda \pi l^2 p p^i}\right) 
    & \geq & \int_0^{\infty} p^x \cdot \left(1 -
    e^{-2\lambda \pi l^2 p p^x}\right) \cdot \left(1 -
    e^{-\lambda \pi l^2 p p^x}\right) dx \\
    & = & \frac{1}{\ln(1/p)} \cdot \left(1 - \frac{1 - e^{-2\lambda
      \pi l^2 p}}{2\lambda \pi l^2 p} - \frac{1 - e^{-\lambda
      \pi l^2 p}}{\lambda \pi l^2 p} + \frac{1 - e^{-3\lambda
      \pi l^2 p}}{3\lambda \pi l^2 p} \right).
\end{eqnarray*}

Note that 
\[ \frac{1 - e^{-3\lambda   \pi l^2 p}}{3\lambda \pi l^2 p} 
	\geq \frac{1 - e^{-2\lambda  \pi l^2 p}}{2\lambda \pi l^2 p} \]
	
which gives us the recursion
\begin{equation}
\label{eq:recursion-final}
\fbox{$\displaystyle \rec(l)  \geq  \frac{(1-3e^{-2})(1-p^2)}{(\lambda \pi l^2 p)^2} + 
\frac{\theta}{4 \pi} \cdot \frac{1-p}{\ln(1/p)} \cdot \left( 1 - \frac{1 - e^{-\lambda
      \pi l^2 p}}{\lambda \pi l^2 p}  \right) \cdot 
 \rec(l - \delta).$}
\end{equation}

Consider the function
\[ \alpha(p) = \frac{(1-p)}{\ln(1/p)}\cdot \left( 1 - \frac{1 - e^{-\lambda
      \pi l^2 p}}{\lambda \pi l^2 p}  \right)\]

For $\alpha(p)$ to be greater than some constant $\eta > 0$ we require
\begin{equation}
\label{eq:eta-condition}
\lambda \pi l^2 p \left( 1-\eta \cdot \frac{\ln(1/p)}{1-p} \right) \geq 1-e^{-\lambda
      \pi l^2 p}.
\end{equation}
Under the condition that
\begin{equation}
\label{eq:eta-constraint}
\eta < \frac{1-p}{\ln(1/p)}.
\end{equation}
we get that~(\ref{eq:eta-condition}) is satisfied whenever
\[ \lambda \pi l^2 p \left( 1-\eta \cdot \frac{\ln(1/p)}{1-p} \right) \geq 1. \]
We denote the least value of $l$ for which this inequality is true by
$\gamma(\eta)$ i.e.
\begin{equation}
\label{eq:eta-value}
\gamma(\eta) = \sqrt{\frac{1}{\lambda \pi p \cdot \left(1-\eta
    \cdot \frac{\ln(1/p)}{1-p} \right)} }.
\end{equation}

Now, in order to lower bound the probability of $u$ and $v$ having a
$\theta$-good path, we choose an $\eta> 0$
satisfying~(\ref{eq:eta-constraint}) and unfold the recursion obtained
in~(\ref{eq:recursion-final}) till the path either connects directly
to $v$ or reaches a point $w$ such that $d(w,v) \leq
\gamma(\eta)$. 
Observing that $\frac{(1-3e^{-2})(1-p^2)}{(\lambda \pi l^2 p)^2} <
\frac{(1-3e^{-2})(1-p^2)}{(\lambda \pi (l-\delta)^2 p)^2}$, and that
the number of edges we may have in the path is finite but unbounded we
get a geometric series sum when we open the recursion, giving us
\begin{equation}
\label{eq:l-star}
\rec(l) \geq \frac{(1-3e^{-2})(1-p^2)}{(\lambda \pi l^2 p)^2} \cdot
\frac{1}{1 - \frac{\eta \theta}{2 \pi} } \cdot \rec(l^*),
\end{equation}
where $l^*
\leq \gamma(\eta)$. Let us assume that the point we have reached is
$w$ with $l^* = d(u,w)<\gamma(\eta)$.  From this point we only
consider the case that $w$ connects directly to $v$. Using the lower
bound of~(\ref{eq:direct}) we get that 
\[ \rec(l^*) \geq \rec_{\direct}(l^*) \geq \frac{(1-3e^{-2})(1-p^2)}{(\lambda \pi
  l^{*2} p)^2} \geq \frac{(1-3e^{-2})(1-p^2)}{(\lambda \pi
  \gamma(\eta)^2 p)^2}.\] 
Substituting this in~(\ref{eq:l-star}) we get
\[ \rec(l) \geq \frac{(1-3e^{-2})(1-p^2)}{(\lambda \pi \gamma(\eta)^2 p)^2}
\cdot \frac{(1-3e^{-2})(1-p^2)}{(\lambda \pi l^2 p)^2} \cdot
\frac{1}{1 - \frac{\eta\theta}{2 \pi}. }
\] 
Substituting the value of $\gamma(\eta)$ from~(\ref{eq:eta-value}) we
get
\[ \rec(l) \geq \left( 1-\eta \cdot \frac{\ln(1/p)}{1-p} \right)^2
\cdot \frac{(1-3e^{-2})^2(1-p^2)^2}{(\lambda \pi l^2 p)^2} \cdot
\frac{1}{1 - \frac{\eta\theta}{2 \pi} }.
\]

\end{proof}

\section{Hierarchical neighbor graphs as a clustering mechanism for WSNs}
\label{sec:sensors}
In this section we simulate the situation where a hierarchical
neighbor graph architecture is used as a clustering mechanism to
collect data from a field of ad hoc wireless sensor devices. We
compare the performance of our architecture to that of
LEACH~\cite{heinzelmann-twc:2002} under assumptions similar to theirs
and find that hierarchical neighbor graphs compare favorably.

The quality of a clustering architecture that collects data from
wireless sensor networks is measured by determining: (a) Data
throughput.  In certain sensor networks applications, the amount of
raw data that is communicated may not be as important as the
information content, or the {\em effective data}, since the data from
certain nodes may be redundant. Hence we consider effective data
throughput as measure of network performance. (b) Energy efficiency
i.e. the energy expended in transmitting a certain amount of data to
the base station.  (c) Network Lifetime. Often it is not possible to
recharge node batteries. We therefore study the {\em network lifetime
}, which we define here as the time it takes for all nodes in the
network to die out.


We consider a typical field of wireless sensors with a set of nodes
$V$ in which a large number of sensor nodes continuously sense data,
process it and communicate the information to an external sink or a
base-station (BS). The sensors coordinate among themselves to form a
hierarchical communication network, its architecture determined by
$\hn_p^w(V)$ where the weight function $w$ depends on the battery
power available at a node at any point of time.  Nodes at higher
levels of the hierarchy get depleted of their energy quicker, and so
there is a heterogeneity in the power profile of the network. We
periodically reform the network to distribute the energy load
according to the residual energy at each node. Once the network
topology has been formed, each node sensor constantly monitors its
environment and periodically sends the data up the hierarchy to its
parent. The nodes at the top-most level of $\hn^w_p(V)$ communicate
directly with the BS.

Sensor data is often highy correlated locally so we aggregate data,
reducing the number of messages transmitted to the BS, hence improving
the energy efficiency of the network.  We consider two different
applications of WSNs and the corresponding data aggregation models
used.  (a) {\em Limited Aggregation}: Only data signals from nodes
located close to each other are highly correlated and can be
aggregated into a single signal. We believe that in $\hn^w_p(V)$,
nodes which share a common parent are located close to each other and
hence the data signals are correlated. In this case, a bounded number
of data signals are fused into a single data signal. (b) {\em
  Unlimited Aggregation}: All data signals, irrespective of location
can be fused to get a single signal. This model is valid for
applications in which we are interested in quantities like the
average, min or max of a set of values e.g. radiation level monitoring
in a nuclear plant where the most useful information for the safety of
plant is the maximum value~\cite{rajagopalan-cst:2006}. In this case,
all data signals at a relaying node are fused into a single signal.

\noindent{\em Network setup and operation.} 
We assume that each node has the computational ability to support MAC
protocols and perform the signal processing functions required. To
begin with, nodes $V$ organize themselves into $\hn^w_p(V)$. The
weight function $w$ is given by $\mbox{batt}_u/\mbox{batt}_{th}$,
where $\mbox{batt}_u$ is the residual battery of node $u$ and
$\mbox{batt}_{th}$ is the threshold energy below which a node is
declared dead. All nodes $\{u$ s.t. $\lev(u)=0 \}$, periodically
transmit data to their parents. All other nodes receive data from all
their children, fuse the incoming data signals along with their own data
according to the data aggregation model and transmit to their
parent. The nodes at the topmost level transmit data to the BS. Nodes
$\{u$ s.t. $\lev(u)=0 \}$ switch off to conserve energy when they are
not transmitting. The operation of the network is divided into
rounds. The network is reformed after the end of each round, which is
after a fixed duration.  We asssume collision free traffic in the
network. This requires each node in $\hn_p(V)$ to create a TDMA
schedule and communicate this to their children. In addition the BS
distributes spreading codes for direct-sequence spread spectrum
(DSSS)~\cite{hu-ton:1993} among the nodes in the networks.

\paragraph{Simulations and Discussion}
In this section we simulate $\hn^w_p(V)$ and compare it to
LEACH~\cite{heinzelmann-twc:2002}. We compare the energy consumption,
network lifetime and the  effective data throughput received at the BS. The
effective data throughput is measured by the number of data signals
represented by the aggregated signal received at the BS. $\hn^w_p(V)$
was observed to transmit 8 timesmore effective data than LEACH, when
using unlimited data aggregation, and 2 to 3 times more when using
limited data compression
     
Now we describe the simulation environment and parameters. We mimicked
LEACH's setup in order to compare our structure with
theirs~\cite{heinzelmann-twc:2002}. We initialized a network with 100
sensor nodes, $V$, spread uniformly over a square region of side 100
unit, each having the same energy of $2J$ to start with. We considered a
node to be dead if node energy drops below $0.1J$. We constructed
$\hn^w_p(V)$, with parameter $p=0.5$, periodically at intervals of
$20s$. The BS is located close to $A$ but outside it.  The bandwidth
of the channel was taken to be 1 Mb/s, and each data signal was taken
to be $500$ bytes long with a $25$ byte packet header for each type of
packet. We simulated this scenario and compared our mechanism against
LEACH, observing energy consumption and throughput in the network. The
time for a round was chosen to be $20s$.

In LEACH for a $N=100$ node network with $k=5$ cluster heads,
approximately $N/k=20$ signals are fused into a single signal at each
cluster head, assuming a uniform distribution.  We simulated our setup
for the same 20:1 compression ratio, as well as for a lower
compression ratio of 10:1. In the second data aggregation model, each
node fuses all signals present to a single signal and then transmits
to it's parent.

Nodes at $\{u$ s.t. $\lev(u)=0 \}$ sense data and send to their
respective parents. Nodes at level $i(>1)$ receive data from all
children, process the data and remove redundancy, and forward the
effective data to their parent. This cycle continues till the highest
level when data is sent to the BS. Nodes at the lowest level usually
spend lesser amount of energy, whereas those at higher level will be
quickly depleted of their battery power.

For our simulation we assume a simple model for energy dissipation
(the same model as used to study LEACH
in~\cite{heinzelmann-twc:2002}). To transmit $l$ bits of data over a
distance $d$, the energy dissipated is $E_{Tx}(l,d) =
lE_{elec}+l\epsilon_{fs}d^2$, and to receive a $l$-bit message is
$E_{Rx}(l)=lE_{elec}$. The radio electronics energy $E_{elec}=50$
$nJ/bit$, depends on the coding and spreading of the signal. And the
amplifier constant $\epsilon_{fs}=10$ $pJ/bit/m^2$ depends on the
acceptable bit-error rate. In addition energy is consumed for data
aggregation which is taken to be $E_{DA}=5$ $nJ/bit/signal$. We work
on the assumption that the nodes are capable of controlling their
power in order to vary the transmission radius.
\begin{figure}[htbp]
\begin{center}
\includegraphics[scale=.80]{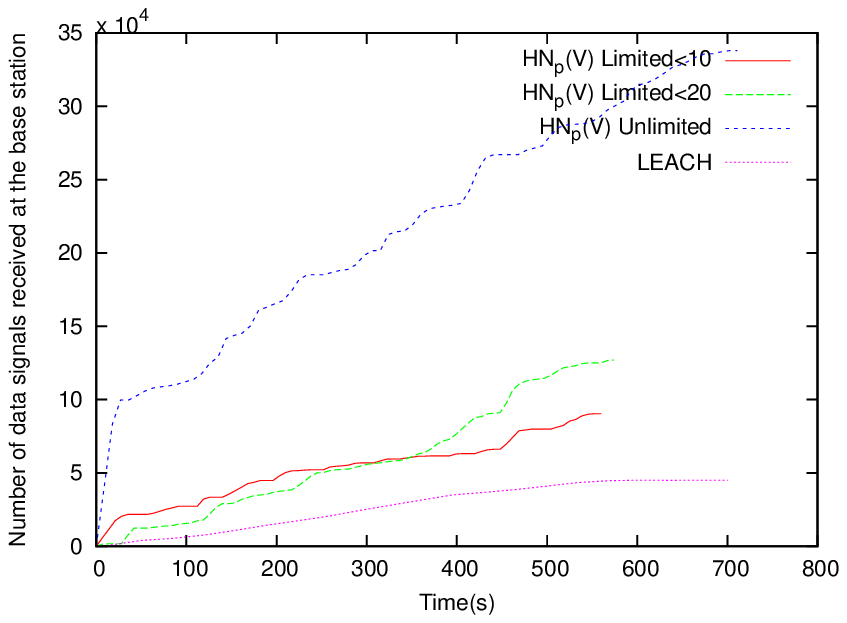}
\caption{Total amount of data received at the BS over time.}
\label{fig:4a}
\end{center}
\end{figure}
\begin{figure}[htbp]
\begin{center}
\includegraphics[scale=.80]{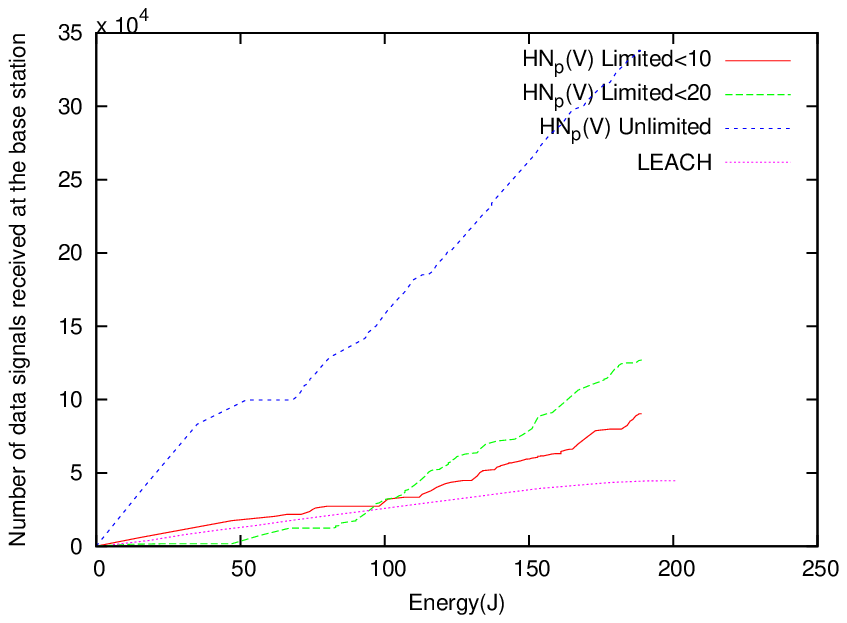}
\caption{Total amount of data received at the BS per given amount of energy.}
\label{fig:4b}
\end{center}
\end{figure}
\begin{figure}[htbp]
\begin{center}
\includegraphics[scale=.80]{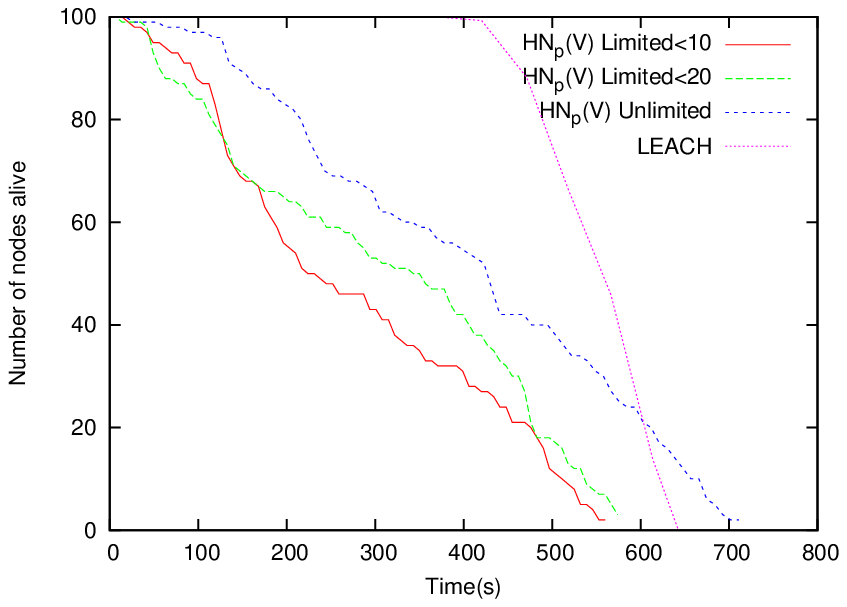}
\caption{Number of nodes alive over time.}
\label{fig:4c}
\end{center}
\end{figure}
Our simulation results are presented in
Figures~\ref{fig:4a},~\ref{fig:4b},~\ref{fig:4c}. The data points for
LEACH are taken from~\cite{heinzelmann-twc:2002}. Clearly when using
unlimited data aggregation (``$\hn_p^w(V)$ Unlimited'' in the plots),
it is observed that the throughput (Figure~\ref{fig:4a}) and the ratio
of data sent to energy consumed (Figure~\ref{fig:4b}) are much higher
than for LEACH, which is expected since our structure aggregates at
several levels. Even when we limit the amount of aggregation allowed
(in the plots ``$\hn_p^w(V)$ Limited $< 20$'' represents a 20:1
compression and ``$\hn_p^w(V)$ Limited $< 10$'' represent 10:1)
$\hn_p^w(V)$ outperforms LEACH in these two aspects. With network
lifetime the picture is a little more complex. In Figure~\ref{fig:4c} we
see that in $\hn_p^w(V)$ the time the last node dies is more or less
the same as LEACH, although LEACH loses a lot of nodes suddenly while
$\hn_p^w(V)$ degrades slowly. In conclusion, it appears from the
simulation that our structure does not dominate LEACH in terms of
lifetime but uses energy more efficiently and provides a higher
throughput than LEACH does, running far ahead in the case of highly
compressible data.

\end{document}

%% file: reactive-flood.pstex_t
\begin{picture}(0,0)%
\epsfig{file=reactive-flood.pstex}%
\end{picture}%
\setlength{\unitlength}{1579sp}%
\begingroup\makeatletter\ifx\SetFigFont\undefined%
\gdef\SetFigFont#1#2#3#4#5{%
  \reset@font\fontsize{#1}{#2pt}%
  \fontfamily{#3}\fontseries{#4}\fontshape{#5}%
  \selectfont}%
\fi\endgroup%
\begin{picture}(10492,4982)(2444,-4258)
\put(5776,-511){\makebox(0,0)[lb]{\smash{{\SetFigFont{7}{8.4}{\rmdefault}{\mddefault}{\updefault}{\color[rgb]{0,0,0}$\conn(\anc_i(v),i)$}%
}}}}
\put(2926,-4186){\makebox(0,0)[lb]{\smash{{\SetFigFont{7}{8.4}{\rmdefault}{\mddefault}{\updefault}{\color[rgb]{0,0,0}$u$}%
}}}}
\put(12301,-3436){\makebox(0,0)[lb]{\smash{{\SetFigFont{7}{8.4}{\rmdefault}{\mddefault}{\updefault}{\color[rgb]{0,0,0}$v$}%
}}}}
\put(7126, 89){\makebox(0,0)[lb]{\smash{{\SetFigFont{7}{8.4}{\rmdefault}{\mddefault}{\updefault}{\color[rgb]{0,0,0}$\anc_i(v)$}%
}}}}
\end{picture}%

%% file: cones.pstex_t
\begin{picture}(0,0)%
\epsfig{file=cones.pstex}%
\end{picture}%
\setlength{\unitlength}{1184sp}%
\begingroup\makeatletter\ifx\SetFigFont\undefined%
\gdef\SetFigFont#1#2#3#4#5{%
  \reset@font\fontsize{#1}{#2pt}%
  \fontfamily{#3}\fontseries{#4}\fontshape{#5}%
  \selectfont}%
\fi\endgroup%
\begin{picture}(5784,4436)(2611,-4527)
\put(4951,-2911){\makebox(0,0)[lb]{\smash{{\SetFigFont{5}{6.0}{\rmdefault}{\mddefault}{\updefault}{\color[rgb]{0,0,0}$u$}%
}}}}
\put(5926,-2236){\makebox(0,0)[lb]{\smash{{\SetFigFont{5}{6.0}{\rmdefault}{\mddefault}{\updefault}{\color[rgb]{0,0,0}$u^1_1$}%
}}}}
\put(7426,-2236){\makebox(0,0)[lb]{\smash{{\SetFigFont{5}{6.0}{\rmdefault}{\mddefault}{\updefault}{\color[rgb]{0,0,0}Cone 1}%
}}}}
\put(7426,-1561){\makebox(0,0)[lb]{\smash{{\SetFigFont{5}{6.0}{\rmdefault}{\mddefault}{\updefault}{\color[rgb]{0,0,0}$u^1_4$}%
}}}}
\put(6601,-1861){\makebox(0,0)[lb]{\smash{{\SetFigFont{5}{6.0}{\rmdefault}{\mddefault}{\updefault}{\color[rgb]{0,0,0}$u^1_3$}%
}}}}
\put(5701,-1561){\makebox(0,0)[lb]{\smash{{\SetFigFont{5}{6.0}{\rmdefault}{\mddefault}{\updefault}{\color[rgb]{0,0,0}$u^1_2$}%
}}}}
\end{picture}%